\documentclass[runningheads]{llncs}
\usepackage[width=1.03\textwidth,height=1.01\textheight,vmarginratio=3:3,hcentering]{geometry}

\usepackage[utf8]{inputenc}
\usepackage[T1]{fontenc}

\usepackage{graphicx}
\usepackage{amsmath}
\usepackage{amssymb}
\usepackage{stmaryrd}
\usepackage{xspace}
\usepackage[disable]{todonotes}
\usepackage{wrapfig}

\newcommand{\tr}{\textrm{\textit{tt}}\xspace}
\newcommand{\fa}{\textrm{\textit{ff}}\xspace}

\usepackage{xcolor}
\usepackage[hidelinks]{hyperref}

\usepackage{listings}

\usepackage{stmaryrd} 
\usepackage{utf8-math}   

\xdefinecolor{maincolor}{RGB}{0, 120, 140}
\colorlet{alertedcolor}{purple}
\colorlet{examplecolor}{green!50!black}

\usepackage{paralist}

\lstdefinestyle{pseudo}{language={},
  basicstyle=\normalfont,
  morecomment=[l]{//},
  morecomment=[s]{/*}{*/},
  morekeywords={for,to,while,do,if,then,else,each,end,Input,Output},
  mathescape=true,
  columns=fullflexible
}

\usepackage{tikz}
\usepackage{tkz-base}
\usepackage{tkz-euclide}
\usetikzlibrary{shapes.misc,calc,decorations,decorations.pathreplacing,backgrounds,shapes.multipart,fit,positioning}

\usepackage{xspace}

\usepackage{booktabs}

\usetikzlibrary{arrows.meta,plotmarks}
\tikzset{
	event/.style={
		draw,
		inner sep=.5pt,
		circle,
		minimum width=6pt
	},
	events/.style={
		font=\tiny,
		xscale=.5,
		yscale=.3
	},
	unit/.style={
		event,
		path picture={ 
			\draw[black]
			(path picture bounding box.south east) -- (path picture bounding box.north west)
			(path picture bounding box.south west) -- (path picture bounding box.north east);
		}
	}
}

\usetikzlibrary{%
  bending,%
  decorations.markings,%
  decorations.pathmorphing}

\makeatletter
\newcounter{streamdiagram@y}
\newenvironment{streamdiagram}[1][]{
  \tikzpicture[streamdiagram,every legend/.style={},#1]
  \setcounter{streamdiagram@y}{0}
  \newcommand{\y}{-\value{streamdiagram@y}}
  \newcommand{\down}[1][1]{\addtocounter{streamdiagram@y}{##1}\renewcommand{\y}{-\value{streamdiagram@y}}}
  \newcommand{\x}[1]{(##1,\y)}
  \newcommand{\legend}[2][]{\down\node[left,every legend,##1] at \x0 {##2};}
  
}{
  \endtikzpicture
}
\makeatother
\pgfdeclarelayer{background}
\pgfsetlayers{background,main}
\pgfkeys{%
  /tikz/on layer/.code={
    \pgfonlayer{#1}\begingroup
    \aftergroup\endpgfonlayer
    \aftergroup\endgroup
  },
}
\tikzset{
  streamdiagram/.style={
    every label/.style={
      inner sep=2pt,
      line width=0pt
    },
    event/.style={
      circle,
      minimum height=10pt,
      minimum width=10pt,
      inner sep=-2pt,
      fill=##1
    },
    event/.default=red!20,
    caption/.style={
      append after command={
        (\tikzlastnode.center) node[font=\scriptsize] {##1}
      }
    }
  }
}

\usepackage{cleveref}

\usepackage{array}
\usepackage{ltlfonts}
\usepackage[ruled,vlined,linesnumbered]{algorithm2e}
\usepackage{arydshln}
\usepackage{wrapfig}

\newcommand{\Spec}{
  \begin{streamdiagram}[yscale=.45,xscale=1]
  \tt
  \legend{In $\LOAD$: Real}

  \legend{Def $\ACC := \ACC[-1|0] + \LOAD[\NOW] - \LOAD[-3|0]$}
  
  \legend{Def $\OK := (\ACC[\NOW] \leq 15)$}
  
\end{streamdiagram}
}

\newcommand{\ExOneLola}{
  \begin{streamdiagram}[yscale=.45,xscale=1]
  \tt
  \legend{}
  \draw[|->, gray] \x0 -- \x5;
  \node[event=blue!20, caption=3] at \x1 {};
  \node[event=blue!20, caption=4] at \x2 {};
  \node[event=blue!20, caption=5] at \x3 {};
  \node[event=blue!20, caption=7] at \x4 {};

  \legend{}
  \draw[|->, gray] \x0 -- \x5;
  \node[event=blue!20, caption=3] at \x1 {};
  \node[event=blue!20, caption=7] at \x2 {};
  \node[event=blue!20, caption=12] at \x3 {};
  \node[event=blue!20, caption=16] at \x4 {};
  
  \legend{}
  \draw[|->, gray] \x0 -- \x5;
  \node[event=green!20, caption=\tr] at \x1 {};
  \node[event=green!20, caption=\tr] at \x2 {};
  \node[event=green!20, caption=\tr] at \x3 {};
  \node[event=red!20, caption=\fa] at \x4 {};
  
\end{streamdiagram}
}

\newcommand{\ExOneAbstractInterpretation}{
\begin{streamdiagram}[yscale=.45,xscale=1]
  \tt
  \legend{}
  \draw[|->, gray] \x0 -- \x5;
  \node[event=blue!20, caption={[1,5]}] at \x1 {};
  \node[event=blue!20, caption=4] at \x2 {};
  \node[event=blue!20, caption=5] at \x3 {};
  \node[event=blue!20, caption=7] at \x4 {};

  \legend{}
  \draw[|->, gray] \x0 -- \x5;
  \node[event=blue!20, caption={[1,5]}] at \x1 {};
  \node[event=blue!20, caption={[5,9]}] at \x2 {};
  \node[event=blue!20, caption={[10,14]}] at \x3 {};
  \node[event=blue!20, caption={[12,20]}] at \x4 {};
  
  \legend{}
  \draw[|->, gray] \x0 -- \x5;
  \node[event=green!20, caption=\tr] at \x1 {};
  \node[event=green!20, caption=\tr] at \x2 {};
  \node[event=green!20, caption=\tr] at \x3 {};
  \node[event=orange!20, caption=?] at \x4 {};
  
\end{streamdiagram}
}
\newcommand{\ExOneSymbolic}{
\begin{streamdiagram}[yscale=.45,xscale=1]
  \tt
  \legend{}
  \draw[|->, gray] \x0 -- \x5;
  \node[event=blue!20, caption={$\LOAD^0$}] at \x1 {};
  \node[event=blue!20, caption=4] at \x2 {};
  \node[event=blue!20, caption=5] at \x3 {};
  \node[event=blue!20, caption=7] at \x4 {};

  \legend{}
  \draw[|->, gray] \x0 -- \x5;
  \node[event=blue!20, caption=$\LOAD^0$] at \x1 {};
  \node[event=blue!20, caption=$\LOAD^0\!\!+\!\!4$] at \x2 {};
  \node[event=blue!20, caption=$\LOAD{}^0\!\!+\!\!9$] at \x3 {};
  \node[event=blue!20, caption=$16$] at \x4 {};
  
  \legend{}
  \draw[|->, gray] \x0 -- \x5;
  \node[event=green!20, caption=\tr] at \x1 {};
  \node[event=green!20, caption=\tr] at \x2 {};
  \node[event=green!20, caption=\tr] at \x3 {};
  \node[event=red!20, caption=\fa] at \x4 {};
  
\end{streamdiagram}
}

\newcommand{\todoml}[1]{\todo[linecolor=cyan,backgroundcolor=cyan!25,bordercolor=cyan]{ML@all: {#1}}}

\newcommand{\R}{\mathbin{\mathcal{R}}}

\newcommand{\Time}{\ensuremath{\mathbb{T}}}
\newcommand{\KWD}[1]{\ensuremath{\textrm{\textit{#1}}}}
\newcommand{\LOAD}{\KWD{ld}}
\newcommand{\ACC}{\KWD{acc}}
\newcommand{\USR}{\KWD{usr}}

\newcommand{\TOTAL}{\KWD{total}}
\newcommand{\OK}{\KWD{ok}}
\newcommand{\NOW}{\KWD{now}}
\newcommand{\BAR}[1]{\ensuremath{\overline{#1}}}
\newcommand{\DD}{\mathbb{D}}
\newcommand{\Dbar}{\BAR{\DD}}
\newcommand{\BB}{\mathbb{B}}
\newcommand{\Bbar}{\BAR{\BB}}

\newcommand{\Expr}{\ensuremath{\textit{Expr}}}

\renewcommand{\And}{\mathrel{\wedge}}

\newcommand{\Into}{\mathrel{\rightarrow}}

\newcommand{\calA}{\mathcal{A}}
\newcommand{\sem}[1]{\llbracket #1 \rrbracket}
\newcommand{\LolaName}{\textrm{Lola}}
\newcommand{\Lola}{\LolaName\xspace}
\newcommand{\LolaLA}{\ensuremath{\LolaName_{\mathcal{LA}}}\xspace}
\newcommand{\LolaB}{\ensuremath{\LolaName_{\mathbb{B}}}\xspace}
\newcommand{\LolaBLA}{\ensuremath{\LolaName_{\mathbb{B}/\mathcal{LA}}}\xspace}

\newcommand{\TRUE}{\ensuremath{\textrm{\textit{tt}}}}
\newcommand{\FALSE}{\ensuremath{\textrm{\textit{ff}}}}

\newcommand{\hannestodo}[1]{\todo[linecolor=red,backgroundcolor=red!25,bordercolor=red]{Hannes: #1}}

\def\orcidID#1{\smash{\href{http://orcid.org/#1}{\protect\raisebox{-1.25pt}{\protect\includegraphics{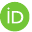}}}}}

\newcommand{\Bool}{\mathbb{B}}
\newcommand{\Real}{\mathbb{R}}

\newcommand{\semSym}[1]{\sem{#1}_\textit{sym}}

\newcommand{\semDen}[1]{\sem{#1}_\textit{den}}

\newcommand{\ite}{\ensuremath{\textit{ite}}}

\newcommand{\Strm}[1]{\mathcal{S}_{#1}}
\newcommand{\Done}{\DD_1}
\newcommand{\DoneP}{\DD'_1}
\newcommand{\Dn}{\DD_n}
\newcommand{\DmP}{\DD'_m}
\newcommand{\calC}{\mathcal{C}}
\newcommand{\calP}{\mathcal{P}}
\newcommand{\semR}[1]{\sem{#1}_{\R}}

\newcommand{\calR}{\mathcal{R}}
\newcommand{\calV}{\mathcal{V}}

\begin{document}
\newcommand{\Thanks}{}

\title{Symbolic Runtime Verification for Monitoring under Uncertainties and Assumptions\Thanks}

\titlerunning{Symbolic Runtime Verification under Uncertainties and Assumptions}

\newcommand{\orcidHannes}{\orcidID{0000-0002-8301-4752}}
\newcommand{\orcidMartin} {\orcidID{0000-0002-3696-9222}}
\newcommand{\orcidCesar} {\orcidID{0000-0003-3927-4773}}


%
%

\author{Hannes Kallwies\inst{1}\orcidHannes \and
  Martin Leucker\inst{1}\orcidMartin \and
  C\'esar S\'anchez\inst{2}\orcidCesar}
\institute{
  University of L\"ubeck, L\"ubeck, Germany
  \and
  IMDEA Software Institute, Madrid, Spain
}
%
%
\maketitle              

\begin{abstract}
  Runtime Verification deals with the question of whether a run of a
  system adheres to its specification.
  This paper studies runtime verification in the presence of partial
  knowledge about the observed run, particularly where input values
  may not be precise or may not be observed at all.
  We also allow declaring assumptions on the execution which permits to
  obtain more precise verdicts also under imprecise inputs.
  To this end, we show how to understand a given correctness property
  as a symbolic formula and explain that monitoring boils down to
  solving this formula iteratively, whenever more and more
  observations of the run are given.
  We base our framework on stream runtime verification, which allows
  to express temporal correctness properties not only in the Boolean
  but also in richer logical theories.
  While in general our approach requires to consider larger and larger
  sets of formulas, we identify domains (including Booleans and Linear
  Algebra) for which pruning strategies exist, which allows to monitor
  with constant memory (i.e. independent of the length of the
  observation) while preserving the same inference power as the
  monitor that remembers all observations.
  We empirically exhibit the power of our technique using a prototype
  implementation under two important cases studies: software for
  testing car emissions and heart-rate monitoring.
  %

\end{abstract}


%

\section{Introduction}
\label{sec:introduction}

In this paper we study runtime verification (RV) for imprecise and
erroneous inputs, and describe a solution---called \emph{symbolic
  monitoring}---that can exploit assumptions about the input and
the system under analysis.
Runtime verification is a dynamic verification technique in which
a given run of a system is checked against a specification, typically
a correctness property~(see
\cite{havelund05verify,leucker09brief,bartocci18lectures}).
In \emph{online monitoring} a monitor is synthesized from the given
correctness property, which attempts to produce a verdict
incrementally from the input trace.
Originally, variants of LTL~\cite{Pnueli77} tailored to finite runs
have been employed in RV to formulate properties (see
\cite{DBLP:journals/logcom/BauerLS10} for a comparison on such
logics).
However, since RV requires to solve a variation of the word problem
and not the harder model-checking problem, richer logics than LTL have
been proposed that allow richer data and
verdicts~\cite{DBLP:journals/sttt/DeckerLT16,DBLP:journals/sttt/HavelundP21}.
\Lola~\cite{DAngeloSSRFSMM05} proposes \emph{stream runtime
  verification} (SRV) where monitors are described declaratively and
compute output streams of verdicts from inputs streams (see
also~\cite{DBLP:conf/sac/LeuckerSS0S18,DBLP:journals/sttt/GorostiagaS21}).
The development of this paper is based on \Lola.

\begin{figure}[t!]
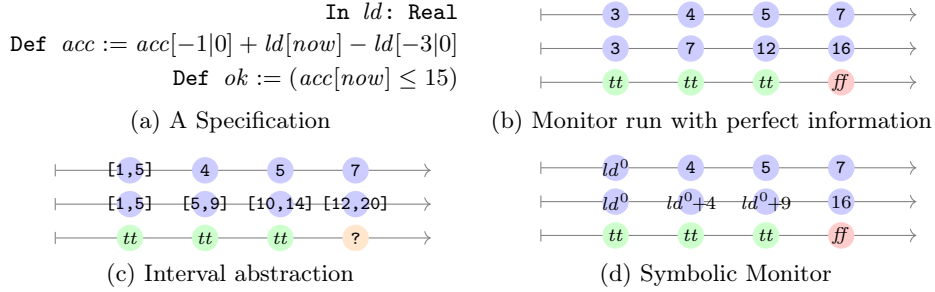

  \begin{tabular}{c@{\hspace{1em}}c}
    \Spec & \ExOneLola \\
    (a) A Specification & (b) Monitor run with perfect information \\[0.6em]
    \ExOneAbstractInterpretation & \ExOneSymbolic \\
    (c) Interval abstraction & (d) Symbolic Monitor \\
  \end{tabular}
  \label{fig:exone}
  \caption{An example specification (a) and three monitors: (b) with
    perfect observability, (c) with an interval abstract domain, (d) a
    symbolic monitor developed in this paper. The symbolic monitor is enriched with
    the additional constraint that $1 \leq \LOAD^0 \leq 5$.}
\end{figure}

\begin{example}
  \label{ex:one}
  Fig.~\ref{fig:exone}(a) shows a \Lola specification with $\LOAD$ as
  input stream (the load of a CPU), $\ACC$ as an output stream that
  represents the accumulated load, computed by adding the current
  value of $\LOAD$ and subtracting the third last value.
  Finally, $\OK$ checks whether $\ACC$ is below $15$.
  The expression $\ACC[-1|0]$ denotes the value of $\ACC$ in the
  previous time point and $0$ as default value if no previous time
  point exists.

  Such a specification allows a direct evaluation strategy whenever
  values on the input streams arrive.
  If, for example, $\LOAD=3$ in the first instant, $\ACC$ and $\OK$
  evaluate to $3$ and $\tr$, respectively.
  Reading subsequently $4, 5, 7$ results in $7,12, 16$ for $\ACC$ and a
  violation is identified on stream $\OK$.
  %
  %
  This is shown in Fig.~\ref{fig:exone}(b).

  A common obstacle in runtime verification is that in practice
  sometimes input values are not available or not given precisely, due
  to errors in the underlying logging functionality or technical
  limitations of sensors.
  In Fig.~\ref{fig:exone}(c) the first value on $\LOAD$ is not
  obtained (but we assume that all values of $\LOAD$ are between $1$
  and $5$).
  One approach is to use interval arithmetic, which can be easily
  encoded as a rich domain in \Lola, and continue the computation
  when obtaining $4$, $5$ and $7$.
  However, at time $4$ the monitor cannot know for sure whether $\OK$
  has been violated, as the interval $[12,20]$ contains $15$.
  This approach based on abstract interpretation~\cite{CousotC77}
  was pursued in \cite{DBLP:conf/rv/LeuckerSS0T19} and suffers from
  this limitation.
  If the unknown input on $\LOAD$ is denoted symbolically by $\LOAD^0$
  %
  we still deduce that $\OK$ holds at time points $1$ to $3$.
  For time point $4$, however, the symbolic representation
  $\ACC^4 = \ACC^3 + 7 - \LOAD^0 = \LOAD^0 + 9 + 7 - \LOAD^0 = 16$
  allows to infer that $\OK$ is clearly violated!
  This is shown in Fig.~\ref{fig:exone}(d).\qed
\end{example}

%
 %

\noindent Example~\ref{ex:one} illustrates the first insight pursued in
this paper: that \emph{Symbolic monitoring} is more precise than
monitoring using abstract domains.

Clearly, an infinite symbolic unfolding of a specification and all the
assumptions with subsequent deduction is practically
infeasible.
Therefore, we perform online monitoring unfolding the specification as
time increases.
We show that a monitor based on deducing verdicts using this partial
symbolic unfolding is both \emph{sound} and \emph{perfect} in the
sense that the monitor only produces correct verdicts and these are as
precise verdicts as possible with the information provided.
Symbolic monitoring done in this straightforward way, however, comes
at a price: the unfolded specification grows as more unknowns and
their inter-dependencies become part of the symbolic unfolding.
For example, in the run in Fig.~\ref{fig:exone}(d) as more unknown
$\LOAD$ values are received, more variables $\LOAD^i$ will be
introduced, which may make the size of the symbolic formula dependent on
the trace length.
We show that for certain logical theories, the current verdict may be
still be computed even after summarizing the history into a compact
symbolic representation, whose size is independent of the trace
length.
For other theories, preserving the full precision of symbolic
monitoring requires an amount of memory that can grow with the trace
length.
More precisely, we show that for the theories of Booleans and of
Linear Algebra, bounded symbolic monitors exist while this is not the
case for the combined theory, which is the second insight presented in this paper.

We have evaluated our approach on two realistic cases studies, the Legal
Driving Cycle \cite{KohlHB18,DBLP:conf/tacas/BiewerFHKSS21} and an ECG
heartbeat analysis (following the \Lola implementation
in~\cite{GorostiagaS21}, see also~\cite{pan85qrs,sznajder17qrs}) which
empirically validate our symbolic monitoring approach, including
constant monitoring on long traces.
When intervals are given for unknown values, our method provides
precise answers more often than previous approaches based on interval
domains~\cite{DBLP:conf/rv/LeuckerSS0T19}.
Especially in the ECG example, these methods are unable to recover
once the input is unknown for even a short time, but our symbolic
monitors recover and provide again precise results, even when the
input was unknown for a larger period.

\paragraph{Related Work.}
Monitoring LTL for traces with mutations (errors) is studied
in~\cite{kauffman20channels} where properties are classified according
to whether monitors can be built that are resilient against the
mutation.
However,~\cite{kauffman20channels} only considers Boolean verdicts and
does not consider assumptions.
The work in~\cite{DBLP:conf/rv/LeuckerSS0T19} uses abstract
interpretation to soundly approximate the possible verdict values when
inputs contain errors for the SRV language
TeSSLa~\cite{ConventHLS0T18}.

Calculating and approximating the values that programs compute is
central to static analysis and program verification.
Two traditional approaches are symbolic execution~\cite{King76} and
abstract interpretation~\cite{CousotC77} which frequently require
over-approximations to handle loops.
In monitoring, a step typically does not contain loops, but the set of
input variables (unlike in program analysis) grows.
Also, a main concern of RV is to investigate monitoring algorithms
that are guaranteed to execute with constant resources.
Works that incorporate assumptions when monitoring include
\cite{DBLP:conf/rv/HenzingerS20,DBLP:conf/rv/CimattiTT21,DBLP:conf/rv/Leucker12}
but uncertainty is not considered in these works, and verdicts are
typically Boolean.
Note that bounded model checking~\cite{BiereCCSZ03} also considers
bounded unfoldings, but it does not solve the problem of building
monitors of constant memory for successive iterations.

%
%
%
%
%

In summary, our contributions are:
\begin{inparaenum}[(1)]
\item A symbolic monitoring algorithm that dynamically unfolds the
  specification, collects precise and imprecise input readings, and
  instantiates assumptions generating a conjunction of formulas.  This
  representation can be used to deduce verdicts even under
  uncertainty, to precisely recover automatically for example under
  windows of uncertainty, and even to anticipate verdicts.
\item A pruning method for certain theories (Booleans and Linear
  Algebra) that guarantees bounded monitoring preserving the power
  to compute verdicts.
\item A prototype implementation and empirical evaluation on realistic case studies.
\end{inparaenum}
%

All missing proofs appear in the appendix.

\section{Preliminaries}
\label{sec:preliminaries}

We use \Lola~\cite{DAngeloSSRFSMM05} to express our monitors.
%
%
\Lola uses first-order sorted theories to build expressions.
These theories are interpreted in the sense that every symbol is a
both constructor to build expressions, and an evaluation function that
produces values from the domain of results from values from the
domains of the arguments.
All sorts of all theories that we consider include the $=$ predicate.

A synchronous stream $s$ over a non-empty data domain $\DD$ is a
function $s: \Strm{\DD} := \mathbb{T} \to \DD$ assigning a value of
$\DD$ to every element of $\mathbb{T}$ (timestamp).
We consider infinite streams ($\mathbb{T} = \mathbb{N}$) or finite
streams with a maximal timestamp $t_{\textrm{max}}$
($\mathbb{T} = [0\ldots{}t_{\textrm{max}}]$).
For readability we denote streams as sequences, so
$s = \langle 1,2,4 \rangle$ stands for
$s: \{1,2,3\} \rightarrow \mathbb{N}$ with
$s(0) = 1, s(1) = 2, s(2) = 4$.
A \Lola specification describes a transformation from a set of input
streams to a set of output streams.

\paragraph{Syntax.}

A \Lola specification $\varphi = (I,O,E)$ consists of a set $I$ of typed
variables that denote the input streams, a set $O$ of typed
variables that denote the output steams, and $E$ which assigns to
every output stream variable $y \in O$ a defining expression $E_y$.
%
The set of expressions over $I\cup O$ of type $\DD$ is denoted by
$E_\DD$ and is recursively defined as:
\( E_\DD = c \mid s[o|c] \mid f(E_{\DD_1}, ..., E_{\DD_n}) \mid ite(E_\mathbb{B}, E_\DD, E_\DD)
\),
where $c$ is a constant of type $\DD$, $s \in I\cup O$ is a
stream variable of type $\DD$, $o \in \mathbb{Z}$ is an offset and $f$
a total function
$\DD_1 \times \dots \times \DD_n \rightarrow
\DD$ ($\ite$ is a special function symbol to denote if-then-else).
The intended meaning of the offset operator $s[o|c]$ is to represent
the stream that has at time $t$ the value of stream $s$ at $t+o$, and
value $c$ used if $t+o\notin\Time$.
A particular case is when the offset is $o=0$ in which case $c$ is not
needed, which we shorten by $s[\NOW]$.
Function symbols allow to build terms that represent complex
expressions.
The intended meaning of the defining equation $E_y$ for output
variable $y$ is to declaratively define the values of stream $y$ in
terms of the values of other streams.

\paragraph{Semantics.}
The semantics of a \Lola specification $\varphi$ is a mapping from
input to output streams.
Given a tuple of concrete input streams
($\Sigma = (\sigma_1, \dots, \sigma_n) \in \Strm{\DD_1} \times \dots
\times \Strm{\DD_n}$) corresponding to input stream identifier
$s_1, \dots, s_n$ and a specification $\varphi$ the semantics of an
expression
$\sem{\cdot}_{\Sigma, \varphi}: E_\DD \rightarrow \Strm{\DD}$ is
iteratively defined as:
\begin{compactitem}
 \item $\sem{c}_{\Sigma, \varphi}(t) = c$
 \item $\sem{s[o|c]}_{\Sigma, \varphi}(t) = \begin{cases}
     \sigma_i(t + o) & \textrm{if } t + o \in \Time \textrm{ and } s = s_i \in I \textrm{ (input stream)}\\
     \sem{e}_{\Sigma, \varphi}(t + o) & \textrm{if } t + o \in \Time \textrm{ and }  E_s = e \textrm{ (output stream)}\\
     c & \textrm{otherwise}
   \end{cases}
$
\item
  $\sem{f(e_1,...,e_n)}_{\Sigma, \varphi}(t) = f(\sem{e_1}_{\Sigma, \varphi}(t),\ldots, \sem{e_n}_{\Sigma, \varphi}(t))$
 \item $\sem{ite(e_1, e_2, e_3)}_{\Sigma, \varphi}(t) = \begin{cases}
     \sem{e_2}_{\Sigma, \varphi}(t) & \textrm{if } \sem{e_1}_{\Sigma, \varphi}(t)  = \tr\\
   \sem{e_3}_{\Sigma, \varphi}(t) & \textrm{if } \sem{e_1}_{\Sigma, \varphi}(t)  = \fa
 \end{cases}$
\end{compactitem}

\noindent The semantics of $\varphi$ is a map
($\sem{\varphi}: (\Strm{\DD_1} \times \dots \times \Strm{\DD_n})
\rightarrow (\Strm{\DD'_1} \times \dots \times\Strm{\DD'_m}$) defined
as
\( \sem{\varphi} (\sigma_1,...,\sigma_n) = (\sem{e'_1}_{\Sigma, \varphi},\ldots, \sem{e'_m}_{\Sigma, \varphi})
\).
%
The evaluation map $\sem{\cdot}_{\Sigma,\varphi}$ is well-defined if
the recursive evaluation above has no cycles.
This acyclicity can be easily checked statically 
(see~\cite{DAngeloSSRFSMM05}).

In online monitoring monitors receive the values incrementally.
The \emph{very efficiently monitorable} fragment of \Lola consists of
specifications where all offsets are negative or 0 (without transitive
0 cycles).
%
%
It is well-known that the very efficiently monitorable specifications
(under perfect information) can be monitored online in a trace length
independent manner.
 In the rest of the paper we also assume that all \Lola specifications
 come with $-1$ or $0$ offsets.
 Every specification can be translated into such a normal form by
 introducing additional streams (flattening).


In this paper we investigate online monitoring under uncertainty for
three special fragments of \Lola, depending on the data theories used:
\begin{compactitem}
\item \textbf{Propositional Logic ($\LolaB$):} The
  data domain of all streams is the Boolean domain
  $\DD = \mathbb{B} = \{\tr, \fa\}$ and available functions are
  $\land, \neg$.
\item \textbf{Linear Algebra
    ($\LolaLA$):}
  The data domain of all streams are real numbers
  $\DD = \mathbb{R}$ and every stream definition has the form
  $c_0 + c_1 * s_0[o_1|d_1] + \dots + c_n * s_n[o_n|d_n]$ where $c_i$ are constants.
\item \textbf{Mixed ($\LolaBLA$):} The
  data domain is $\mathbb{B}$ or $\mathbb{R}$. Every stream definition
  is either contained in the Propositional Logic fragment extended by
  the functions $<, \leq, =$ or in the Linear Algebra fragment.
\end{compactitem}

\section{A Framework for Symbolic Monitoring}
\label{sec:framework}

In this section we introduce a general framework for monitoring using
symbolic computation, where the specification and the information
collected by the monitor (including assumptions and precise and
imprecise observations) are presented symbolically.

%
%

\subsection{Symbolic Expressions}
\label{subsec:exprt}

Consider a specification $\varphi=(I,O,E)$.
We will use symbolic expressions to capture the relations between the
different streams at different points in time.
We introduce the \emph{instant variables} $x^t$ for a given stream
variable $x\in I\cup{}O$ and instant $t\in\Time$.
The type of $x^t$ is that of $x$.
Considering Example~\ref{ex:one}, $\LOAD^3$ represents the real
value that corresponds to the input stream $\LOAD$ at instant
$3$ which is $7$.
The set of instant variables is $V=(I\cup O)\times \Time$.

\begin{definition}[Symbolic Expression]
  Let $\varphi$ be a specification and $\calA$ a set of variables that
  contains all instant variables (that is $V\subseteq \calA$), the set
  of symbolic expressions $\Dbar$ is the smallest set containing (1)
  all constants $c$ and all symbols in $a \in \calA$, (2) all
  expressions $f(t_1,\ldots,t_n)$ where $f$ is a constructor symbol
  of type $\DD_1 \times \dots \times \DD_n \rightarrow \DD$ and
  $t_i$ are elements of $\Dbar$ of type $\DD_i$.
\end{definition}

\noindent We use $\Expr^\DD_\varphi(\calA)$ for the set of symbolic expressions of type $\DD$
(and drop $\varphi$ and $\calA$ when it is clear from the context).

\begin{example}
  Consider again Example~\ref{ex:one}.
  The symbolic expression $\ACC^3+\LOAD^4$, of type $\Real$,
  represents the addition of the load at instant $4$ and the
  accumulator at instant $3$.
  Also, $\ACC^4=\ACC^3+\LOAD^4$ is a predicate (that is, a $\Bool$
  expression) that captures the value of $\ACC$ at instant $4$.
  %
  %
  The symbolic expression $\LOAD^1=4$ corresponds to the reading of
  the value $4$ for input stream $\LOAD$ at instant $1$.
  Finally, $1\leq\LOAD^0 \And \LOAD^0\leq 5$ corresponds to the
  assumption at time $0$ that $\LOAD$ has value between $1$ and
  $5$. \qed
\end{example}

\subsection{Symbolic Monitor Semantics}
\label{subsec:symsem}

We define the symbolic semantics of a \Lola specification
$\varphi=(I,O,E)$ as the expressions that result by instantiating the
defining equations $E$.

 \begin{definition}[Symbolic Monitor Semantics]
   \label{def:symmonsem}
   The map $\sem{\cdot}_\varphi:E_\DD\Into\mathbb{T}\Into\Expr^\DD_\varphi$ is
   defined as $\sem{c}_\varphi(t)=c$ for constants, and
   \begin{compactitem}
   \item 
     $\sem{f(e_1,\ldots,e_n)}_\varphi(t) =
     f(\sem{e_1}_\varphi(t),\ldots, \sem{e_n}_\varphi(t))$
   \item $\sem{ s[o|c]}_\varphi(t) = s^{t+o}$ if
     $t + o \in \mathbb{T}$, or $\sem{ s[o|c]}_\varphi(t) = c$   otherwise.
   \end{compactitem}
   The symbolic semantics of a specification $\varphi$ is the map
   $\semSym{\cdot}:\Time\Into2^{\Expr_\varphi^\Bool}$ defined as
   \( \semSym{\varphi}^t=\{ y^t=\sem{E_y}_\varphi(t) \;|\; \text{ for
     every $y\in O$}\}.  \)
 \end{definition}


 A slight modification of the symbolic semantics allows to obtain
 equations whose right hand sides only have input instant variables:
 \begin{compactitem}
 \item  $\sem{ s[o|c]}_\varphi(t) = s^{t+o}$ \hspace{2.7em}if  $t + o \in \mathbb{T}$ and $s\in I$
 \item $\sem{ s[o|c]}_\varphi(t) = \sem{e_s}(t+o)$ if  $t + o \in \mathbb{T}$  and $s\in O$
 \item $\sem{ s[o|c]}_\varphi(t)=c$ otherwise
 \end{compactitem}
 %
We call this semantics the symbolic unrolled semantics, which
corresponds to what would be obtained by performing equational
reasoning (by equational substitution) in the symbolic semantics.
\begin{example}
  \label{ex:symsem}
  Consider again the specification $\varphi$ in Example~\ref{ex:one}.
  The first four elements of $\semSym{\varphi}$ are (after
  simplifications like $0+x = x$ etc.):

  \noindent%
  \[
    \footnotesize
    \begin{array}{|c|c|c|c|}\hline
    0 & 1 & 2 & 3 \\ \hline
    \ACC^0=\LOAD^0 & \ACC^1=\ACC^0+\LOAD^1 &\ACC^2=\ACC^1+\LOAD^2 &\ACC^3=\ACC^2+\LOAD^3-\LOAD^0 \\
      \OK^0=\ACC^0\leq 15 & \OK^1=\ACC^1\leq 15 & \OK^2=\ACC^2\leq 15 & \OK^3=\ACC^3\leq 15 \\ \hline
  \end{array}
\]

Using the unrolled semantics the equations for $\OK$ would be, at time
$0$, $\OK^0=\LOAD^0\leq 15$, and at time $1$, $\OK^1=\LOAD^0 + \LOAD^1\leq 15$.
In the unrolled semantics all equations contain only instant
variables that represent inputs.\qed
\end{example}

Recall that the denotational semantics of \Lola monitor specifications
in Section~\ref{sec:preliminaries} maps every tuple of input streams
into a tuple of output streams, that is
\(
  \sem{\varphi}:\Strm{\Done}\times\dots\times\Strm{\Dn}\Into\Strm{\DoneP}\times\dots\times\Strm{\DmP}.
\)
The symbolic semantics also has a denotational meaning even without
receiving the input stream, defined as follows.

%
\begin{definition}[Denotational semantics]
  Let $\varphi=(I,O,E)$ be a specification with $I = (x_1,\ldots,x_n)$
  and $O = (y_1,\ldots,y_m)$.  The denotational semantics of a set of
  equations $E\subseteq\Expr^\Bool_\varphi$
  $\semDen{E}\subseteq\Strm{\Done}\times\dots\times\Strm{\Dn}\times\Strm{\DoneP}\times\dots\times\Strm{\DmP}$
  is:
\begin{align*}
  \semDen{E} & =\{ (\sigma_1,\ldots,\sigma_n,\sigma_1',\ldots,\sigma'_m)\;|\; \text{ for every $e\in E$ } \\
  & \hspace{3em}\{x_1^t=\sigma_1(t), \ldots, x_n^t=\sigma_n(t),y_1^t=\sigma'_1(t),\ldots,y_m^t=\sigma'_m(t)\} \models e \}
\end{align*}
\end{definition}

\todo{What is the meaning of this for
  $|\mathbb{T}|=\infty$?}\todoml{$t$ is always a finite value and for
  any $t$ it should work. If we have a continuous setting, we may say
  something regarding the fixpoint $t \to \infty$ but we don't care
  here. In other words: Yes, we do not discuss that the semantics for
  $\infty$ is say $\tr$ but the monitor will always say $\fa$. We do
  not support full anticipation. Cesar: Idea for future paper.}

\noindent Using the previous definition,
$\semDen{\bigcup_{i\leq t}\semSym{\varphi}^i}$ corresponds to all
the tuples of streams of inputs and outputs that satisfy the
specification $\varphi$ up to time $t$.



\subsubsection{A Symbolic Encoding of Inputs, Constraints and Assumptions.}
Input readings can also be defined symbolically as follows.
Given an instant $t$, an input stream variable $x$ and a value $v$,
the expression $x^t=v$ captures the precise reading of $v$ at 
$t$ on $x$.
Imprecise readings can also be encoded easily.
For example, if at instant $3$ an input of value $7$ for $\LOAD$ is
received by a noisy sensor (consider a $1$ unit of tolerance), then
$6\leq\LOAD^3\leq 8$ represents the imprecise reading.

Assumptions are relations between the variables that we assume to hold
at all positions, which can be encoded as stream expressions of type
$\Bool$.
For example, the assumption that the load is always between $1$ and
$10$ is $1\leq\LOAD[\NOW]\leq 10$.
Another example, $\LOAD[-1|0]+1\geq\LOAD[\NOW]$ which encodes that
$\LOAD$ cannot increase more than $1$ per unit of time.
We use $A$ for the set of assumptions associated with a \Lola
specification $\varphi$ (which are a set of stream expressions of type
$\Bool$ over $I\cup O$).

\subsection{A Symbolic Monitoring Algorithm}



%
Based on the previous definitions we develop our symbolic monitoring
algorithm shown in Alg.~\ref{alg:symonline}.
Line~3 instantiates the new equations and assumptions from the
specification for time $t$.
Line~4 incorporates the readings (perfect or imperfect).
Line~5 performs evaluations and simplifications, which is dependent on
the particular theory.  In the case of past-specifications with
perfect information this step boils down to substitution and evaluation. Line~$6$
produces the output of the monitor.
\renewcommand{\algorithmcfname}{Alg.}
\begin{wrapfigure}[11]{l}{0.54\textwidth}
  \begin{minipage}{0.54\textwidth}
    \vspace{-2.3em}
    \begin{algorithm}[H] 
      
      $t\leftarrow 0$ and $E\leftarrow \emptyset$\;
      \While{$t\in \Time$}{
        $E\leftarrow E\cup \semSym{\varphi}^t \cup \sem{A^t}_\varphi$\;
        \mbox{$E\leftarrow E\cup \{ x^t=v \;|\; \text{ for inputs $x$} \}$}\;
        Evaluate and Simplify\;
        Output\;
        Prune\;
        $t\leftarrow t + 1$ \;
      }
  \caption{\mbox{Online Symbolic Monitor for $\varphi$}}
  \label{alg:symonline}
\end{algorithm}
\end{minipage}
\end{wrapfigure}
Again, this is application dependent.
In the case of past specifications with perfect information the output
value will be computed without delay and emitted in this step.
In the case of $\Bool$ outputs with imperfect information, an SMT solver
can be used to discard a verdict.
For example, to determine the value of $\OK$ at time $t$, the verdict
$\TRUE$ can be discarded if $\exists *.ok^t$ is UNSAT, and the verdict
$\FALSE$ can be discarded if $\exists *.\neg ok^t$ is UNSAT.
For richer domains specific reasoning can be used, like emitting lower
and upper bounds or the set of constraints deduced.
Finally, Line~7 eliminates constraints that will not be necessary for
future deductions and performs variable renaming and summarization to
restrict the memory usage of the monitor (see
Section~\ref{sec:atwork}).
For past specifications with perfect information, after step $5$ every
equation will be evaluated to $y^t=v$ and the pruning will remove from
$E$ all the values that will never be accessed again.
Example~\ref{ex:symsem} in Appendix~\ref{app:examples} illustrates how
the algorithm handles imperfect information and pruning for the
specification of Example~\ref{ex:one}.

The symbolic monitoring algorithm generalizes the concrete monitoring
algorithm by allowing to reason about uncertain values, while it still
obtains the same results and performance under certainty.
Concrete monitoring allows to monitor with constant amount of
resources specifications with bounded future references when inputs
are known with perfect certainty.

Symbolic monitoring, additionally, allows to handle uncertainties and
assumptions, because the monitor stores equations that include
variables that capture the unknown information, for example the
unknown input values.
We characterize a symbolic monitor as a step function
$M:2^\Expr_\varphi\Into2^\Expr_\varphi$ that transforms expressions
into expressions.
At a given instant $t$ the monitor collects readings
$\psi^t\in\Expr_\varphi$ about the input values and applies the step function to the previous
information and the new information.
Given a sequence of input readings 
$\psi^1,\psi^2\ldots$ we use $M^0=M(\psi^0)$ and
$M^{i+1}=M(M^i\cup\psi^{i+1})$ for the sequence of monitor states reached
by the repeated applications of $M$.
We use $\Phi^t=\cup_{i\leq t}(\semSym{\varphi}^i\cup\sem{A^i}_\varphi\cup\psi^i)$ for the
formula that represents the unrolling of the specification and the current assumptions together
with the knowledge about inputs collected up-to $t$.
%

\begin{definition}[Sound and Perfect monitoring]
  Let $\varphi$ be a specification, $M$ a monitor for $\varphi$,
  $\psi^1,\psi^2\ldots$ a sequence of input observations,
  and $M^1,M^2\ldots$ the monitor states reached after repeatedly
  applying $M$.
  Consider an arbitrary predicate $\alpha$ involving only instant
  variables $x^t$ at time $t$.
  \begin{compactitem}
  \item $M^t$ is sound if whenever $M^t\models\alpha$ then
    $\Phi^t\models \alpha$.
  \item $M^t$ is perfect if it is sound and if
    $\Phi^t\models \alpha$ then $M^t\models\alpha$.
  \end{compactitem}
  \label{def:SoundPerfectMon}
\end{definition}
Note that soundness and perfectness is defined in terms of the ability
to infer predicates that only involve instant variables at time $t$,
so the monitor is allowed to eliminate, rename or summarize the rest
of the variables.
It is trivial to extend this definition to expressions $\alpha$ that
can use instant variables $x^{t'}$ with $(t-d)\leq{}t'\leq t$ for some
constant $d$.
If a monitor is perfect in this extended definition it will be able to
answer questions for variables within the last $d$ steps.

The version of the symbolic algorithm presented in
Alg.~\ref{alg:symonline} that never prunes (removing line 7) and
computes at all steps $\Phi^t$ is a sound and perfect monitor.
However, the memory that the monitor needs grows without bound if the
number of uncertain items also grows without bound.
In the next section we show that (1) trace length independent perfect
monitoring under uncertainty is not possible in general, even for past
only specifications and (2) we identify concrete theories, namely
Booleans and Linear Algebra and show that these theories allow perfect
monitoring with constant resources under unbounded uncertainty.

\section{Symbolic Monitoring at Work}
\label{sec:atwork}

%

 \medskip
 \refstepcounter{example}
 \noindent \textit{Example \theexample.}
  \label{ex:loadPrune}
  Consider the \Lola specification on the left, where the Real input
  stream $ld$ indicates the current CPU load and the Boolean input stream
  $usr_a$ indicates if the currently active user is user A.
  This specification checks whether the accumulated load of user A is
  at most 50\% of the total accumulated load.
  Consider the inputs
  $\LOAD = \langle ?, 10, 4, ?, ?, 1, 9, \dots \rangle$,
  $\USR_a = \langle \fa, \fa, \fa, \tr, \tr, \tr, \fa, \dots \rangle$
  from $0$ to $6$.
  \begin{wrapfigure}[5]{l}{0.5\textwidth}
    \begin{minipage}{0.5\textwidth}
      \vspace{-2.8em}
      \[
    \begin{array}{rcl}
  \ACC & := & \ACC[-1|0] + ld[\NOW] \\ 
      \ACC_a & := & \ACC_a[-1|0] + ite(\USR_a[\NOW],\\
      &&\hspace{8em}\LOAD[\NOW], 0) \\
  \OK & := & \ACC_a \leq 0.5 * \ACC \\
    \end{array}
  \]
\end{minipage}
\end{wrapfigure}
Also, assume that at every instant $t$, $0 \leq ld^t \leq 10$.
  At instant $6$ our monitoring algorithm would yield the symbolic
  constraints $(acc^6 = 24 + ld^0 + ld^3 + ld^4)$ and
  $(acc_a^6 = 1 + ld^3 + ld^4)$ for $acc^6$ and $acc_a^6$, and the
  additional
  $(0 \leq ld^0 \leq 10\; \land 0 \leq ld^3 \leq 10\; \land\ 0 \leq
  ld^4 \leq 10)$.
  %
  %
  %
  An existential query to an SMT solver allows to conclude that
  $\OK^6$ is true since $\ACC_a^6$ is at most $21$ but then $\ACC^6$ is $44$.
  However, every unknown variable from the input will appear in one of
  the constraints stored and will remain there during the whole
  monitoring process.  \qed

When symbolic computation is used in static analysis, it is not a
common concern to deal with a growing number of unknowns as usually
the number of inputs is fixed a-priori.
In contrast, a goal in RV is to build online monitors that are
trace-length independent, which means that the calculation time and
memory consumption of a monitor stays below a constant bound and does
not increase with the received number of inputs.
%
In Example~\ref{ex:loadPrune} above this issue can be tackled by rewriting the
constraints as part of the monitor's pruning step
using
$n \leftarrow ld^0$, $m \leftarrow (ld^3 + ld^4)$ to obtain
$( acc^t =  24 + n + m)$, $(acc_a^t  =  1 + m)$ and
$(0 \leq n \leq 10) \land (0 \leq m \leq 20)$.
  %
  From the rewritten constraints it can still be deduced that
  $acc_a^6 \leq 0.5 * acc^6$. 
  Note also that every instant variable in the specification only
  refers to previous instant variables.
  Thus for all $t \geq 7$, there is no direct reference to either
  $ld^3$ or $ld^4$. Variables $ld^3$ and $ld^4$ are, individually, no
  longer \emph{relevant} for the verdict and it does not harm to
  denote $ld^3+ld^4$ by a single variable $m$.
  We call this step of rewriting \emph{pruning} (of non-relevant
  variables).
 
  Let $\calC^t \subseteq \Expr^\Bool_\varphi$ be the set of constraints
  maintained by the monitor that encode its knowledge about inputs
  and assumptions for the given specification.
  In general, pruning is a transformation of a set of constraints $\calC^t$ into 
  a new set $\calC'^t$ requiring less memory, but is still describing
  the same relations between the instant variables:
  
  \begin{definition}[Pruning strategy]
    Let $\calC \subseteq \Expr^\Bool$ be a set of propositions
  over variables $\mathcal{A}$ and
  ${\R} = \{r_1, \dots, r_n\} \subseteq \mathcal{A}$ the subset of
  \emph{relevant} variables.
  We use $|\calC|$ for a measure on the size of $\calC$.
  A \emph{pruning strategy}
  $\calP: 2^{\Expr^{\Bool}} \rightarrow 2^{\Expr^{\Bool}}$ is a
  transformation such that for all $\calC\in\Expr^\Bool$,
  $|\calP(\calC)| \leq |\calC|$.
   A Pruning strategy $\calP: 2^{\Bbar} \rightarrow 2^{\Bbar}$ is called
   \begin{compactitem}
   \item \emph{sound}, \hspace{0.1em} whenever for all $\calC \subseteq \Expr^\Bool$,
     $\semR{\calC} \subseteq \semR{\calP(\calC)}$,
   \item \emph{perfect}, whenever for all 
     $\calC \subseteq \Expr^\Bool$,  $\semR{\calC}=\semR{\calP(\calC)}$,
   \end{compactitem}
   where
   $\semR{\calC} = \{(v_1,\dots,v_n) | (r_1 = v_1 \land \dots
 \land r_n = v_n) \models \calC \}$ is the set of all value
 tuples for ${\R}$ that entail the constraint set $\calC$.
 We say that the pruning strategy is \emph{constant} if
 for all $\calC \subseteq \Expr^\Bool: |\calP(\calC)|
 \leq c$ for a constant $c \in \mathbb{N}$.
  \label{def:pruningStrat}
  \end{definition}
  
  A monitor that exclusively stores a set $\calC^t$ for every
  $t \in \mathbb{T}$ is called a constant-memory monitor if there is a
  constant $c \in \mathbb{N}$ such that for all $t$,
  $|\calC^t| \leq c$.

  Previously we defined an online monitor $M$ as a function that
  iteratively maps sets of expressions to sets of expressions.
  Clearly, the amount of information to maintain grows unlimited if we
  allow the monitor to receive constraints that contain information of
  an instant variable at time $t$ at any other time $t'$.
  Consequently, we first restrict our attention to \emph{atemporal
    monitors}, defined as those which receive proposition sets that
  only contain instant variables of the current instant of time.
  Atemporal monitors cannot handle assumptions like
  $ld[-1|0] \leq 1.1 * ld[\NOW]$.
  At the end of this section we will extend our technique to monitors that may refer $n$ instants to the past.

  \newcounter{thm-perfect}
  \setcounter{thm-perfect}{\value{theorem}}
  
  \begin{theorem}
   \label{th:perfectness}
   Given a specification $\varphi$ and a constant pruning strategy
   $\calP$ for $\Expr^\Bool_\varphi$, there is an atemporal
   constant-memory monitor $M_\varphi$ s.t.
   \begin{compactitem}
    \item $M_\varphi$ is sound if the pruning strategy is sound. 
    \item $M_\varphi$ is perfect if the pruning strategy is perfect.
   \end{compactitem}
  \end{theorem}

Yet we have not given a complexity measure for our constraint sets.
For our approach we use the number of variables and constants in the constraints, 
that is
$|\calC| = \sum_{\varphi \in \calC} |\varphi|$ and $|c| = 1$,
$|v| = 1$,
$|f(e_1,\dots,e_n)| = |e_1| + \dots + |e_n|, |\textit{ite}(e_1, e_2,
e_3)| = |e_1| + |e_2| + |e_3|$ for a constant $c$ and an atomic
proposition $v$.

\subsection{Application to \Lola fragments}
\label{subsec:PruningLOLAFragments}
  
%
We describe now perfect pruning strategies for $\LolaB$ and $\LolaLA$.
For $\LolaBLA$ we will show that no such perfect pruning strategy
exists but present a sound and constant pruning strategy.


\subsubsection{\LolaB:}
First we consider the fragment $\LolaB$ where all
input and output streams, constants and functions are of type Boolean.
Consequently, constraints given to the monitor only contain variables,
constants and functions of type Boolean.

 \medskip
 \refstepcounter{example}
 \noindent \textit{Example \theexample.}
 Consider the following specification (where all inputs are uncertain,
 $\oplus$ denotes exclusive or) shown on the left.
  The unrolled semantics, shown on the right, indicates that \OK{} is
  always true.

  \noindent\begin{tabular}{l@{\hspace{2.5em}}l}
             $ 
    \begin{array}{rcl}
    \\
  a & := & a[-1|\FALSE] \oplus x[\NOW] \\
  b & := & b[-1|\TRUE] \oplus x[\NOW] \\
  \OK & := & a[\NOW] \oplus b[\NOW] \\
    \end{array}
$            
             &
               \footnotesize
    \begin{tabular}{|p{1.5em}|p{4em}|p{6em}|p{8em}|l}\hline  
    0 & 1 & 2 & 3 & \ldots \\\hline
    $\phantom{\neg}x^0$ & $\phantom{\neg}x^0\oplus x^1$ & $\phantom{\neg}x^0\oplus x^1\oplus x^2$ & $\phantom{\neg}x^0\oplus x^1\oplus x^2 \oplus x^3$ & \ldots \\
    $\neg x^0$ & $\neg x^0\oplus x^1$ & $\neg x^0\oplus x^1\oplus x^2$ & $\neg x^0\oplus x^1 \oplus x^2 \oplus x^3$ &  \ldots \\
    $\TRUE$ & $\TRUE$ & $\TRUE$ & $\TRUE$ & \ldots \\\hline
\end{tabular}
  \end{tabular}
  \vspace{.8em}

%



 \noindent However, the Boolean formulas maintained internally by the
 monitor are continuously increasing.
 Note that at time 1 the possible combinations of $(a^1, b^1, \OK^1)$
 are $(\fa,\tr,\tr)$ and $(\tr,\fa,\tr)$, as shown below (left).
 By eliminating duplicates from
  \begin{wrapfigure}[6]{l}{0.42\textwidth}
    \begin{minipage}{0.42\textwidth}
      \vspace{-2.3em}
 \[
   \begin{array}{@{}l@{\hspace{1.5em}}l}
     \begin{array}{c|cccc}
       (x^0,x^1) &  0 0 & 0 1 & 1 0 & 1 1 \\
       \hline
       a^1 & \FALSE  & \TRUE & \TRUE & \FALSE \\
       b^1 & \TRUE  & \FALSE & \FALSE & \TRUE \\
       \OK^1 & \TRUE & \TRUE & \TRUE & \TRUE \\
     \end{array} &
  \begin{array}{c|cc}
    v^1 &  0 & 1 \\
    \hline
    a^1 & \FALSE  & \TRUE \\
    b^1 & \TRUE & \FALSE \\
    \OK^1 & \TRUE & \TRUE \\
  \end{array}
   \end{array}
 \]
\end{minipage}
\end{wrapfigure}
 this table we obtain another table
 with two columns which can be expressed by formulas over a single,
 fresh variable $v^1$ (as shown on the right).
From this table we can directly infer the new formulas
 $a^1 = v^1$, $b^1 = \neg v^1$, $\OK^1 = \tr$, which preserve the
 condensed information that $a^1$ and $b^1$ are opposites.
 We can use these new formulas for further calculation.
 At time $2$, $a^2 = v^1\oplus x^2$, $b^1 = \neg v^1\oplus x^2$ which
 we rewrite as $a^2 = v^2$, $b^1 = \neg v^2$ again concluding
 $\OK^1 = \tr$.
 %
 This illustrates how the pruning guarantees a constant-memory
 monitor. Note that this monitor will be able to infer at every step
 that $\OK$ is $\TRUE$ even without reading any input. \qed

  
The strategy from the example above can be generalized to a
pruning strategy.
Let $\calR = \{r_1,\dots,r_m\}$ be the set of relevant variables (in
our case the output variables ${s}_i^t$) and
$\mathcal{V} = \{s_1,\dots,s_n\} \cup \calR$ all
variables (in our case input variables and fresh variables from
previous pruning applications).
Let $\calC$ be a set of constraints over $r_1,\dots,r_m,s_1,\dots,s_m$,
which can be rewritten as a Boolean expression $\gamma$ by conjoining
all constraints.

The method generates a value table $T$ which includes as columns all
value combinations of $(v_1,\dots,v_m)$ for $(r_1,\dots,r_m)$ such
that $(r_1 = v_1) \land \dots \land (r_m = v_m) \models \gamma$.
Then it builds a new constraint set $\calC'$ with an expression
$r_i = \psi_i(v_1,\dots,v_k)$ for every $1 \leq i \leq m$ over $k$
fresh variables, where the $\psi_i$ are generated from the rows of the
value table. The number of variables is $k = \lceil \log(c) \rceil$
with $c$ being the number of columns in the table (i.e. combinations
of $r_i$ satisfying $\gamma$).
This method is the $\LolaB$ pruning strategy which is perfect.
By Theorem~\ref{th:perfectness} this allows to build a perfect
atemporal constant-memory monitor for $\LolaB$.

\newcounter{lem-lolaBperfect}
\setcounter{lem-lolaBperfect}{\value{lemma}}

\begin{lemma}
  \label{lem:lolaBperfect}
  The  $\LolaB$ pruning strategy is perfect and constant.
\end{lemma}

 
%
  
\subsubsection{$\LolaLA$:}

The same idea used for \LolaB can be adapted to Linear Algebra.

 \medskip
 \refstepcounter{example}
 \noindent \textit{Example \theexample.}
 Consider the specification on the left.
 The main idea is that $\ACC_a$
\begin{wrapfigure}[5]{l}{0.48\textwidth}
  \begin{minipage}{0.48\textwidth}
    \vspace{-2.5em}
\[
\begin{array}{rcl}
  \ACC_a & := & \ACC_a[-1|0] + \LOAD_a[\NOW] \\
  \ACC_b & := & \ACC_b[-1|0] + \LOAD_b[\NOW] \\
  \TOTAL & := & \TOTAL[-1|0] + \frac{1}{2}(\LOAD_a[\NOW]+\\
  && \hspace{7em}\LOAD_b[\NOW])
\end{array}
\]
\end{minipage}
\end{wrapfigure}
accumulates the load of CPU A (as indicated by $\LOAD_a$), and
similarly $\ACC_b$ accumulates the load of CPU B (as indicated by
$\LOAD_b$).
Then, $\TOTAL$ keeps the average of $\LOAD_a$ and $\LOAD_b$.
The unrolled semantics is

{\footnotesize
\vspace{.8em}
\ \hspace{-1.5em}\begin{tabular}{|p{5em}|p{11em}|p{17em}|l}\hline  
  0 & 1 & 2 & \ldots \\ \hline
  $\LOAD_a^0$ & $\LOAD_a^0 + \LOAD_a^1$ & $\LOAD_a^0 + \LOAD_a^1 + \LOAD_a^2$ & \ldots \\
  $\LOAD_b^0$ & $\LOAD_b^0 + \LOAD_b^1$ & $\LOAD_b^0 + \LOAD_b^1 + \LOAD_b^2$ &  \ldots \\
  $\frac{1}{2}(\LOAD_a^0+\LOAD_b^0)$ & $\frac{1}{2}((\LOAD_a^0+\LOAD_b^0)+(\LOAD_a^1+\LOAD_b^1))$ & $\frac{1}{2}((\LOAD_a^0+\LOAD_b^0)+(\LOAD_a^1+\LOAD_b^1)+(\LOAD_a^2+\LOAD_b^2))$   & \ldots \\ \hline
\end{tabular}
\vspace{.8em}
}

\noindent Again, the formulas maintained during monitoring are
increasing.
The formulas at $0$ cannot be simplified, but at $1$, $\LOAD_a^0$ and
$\LOAD_a^1$ have exactly the same influence on $\ACC_a^1, \ACC_b^1$
and $\TOTAL$.
To see this consider $(\ACC_a^1, \ACC_b^1$, $\TOTAL^1)$ as matrix
multiplication shown below on the left.
The matrix in the middle just contains two linearly independent
vectors.
Hence the system of equations can be equally written as shown in the
right, over two fresh variables $u^1, v^1$:
\[
  \begin{array}{l@{\hspace{5em}}l}
    \left(
    \begin{array}{c}
      \ACC_a^1\\
      \ACC_b^1\\
      \TOTAL^1
    \end{array}
    \right)
    =
    \left(
    \begin{array}{ccccc}
      1 & 0 & 1 & 0\\
      0 & 1 & 0 & 1\\
      \frac{1}{2} & \frac{1}{2} & \frac{1}{2} & \frac{1}{2}\\
    \end{array}
    \right)
    *
    \left(
    \begin{array}{c}
      \LOAD_a^0\\
      \LOAD_b^0\\
      \LOAD_a^1\\
      \LOAD_b^1\\
    \end{array}
    \right) &
\left(
 \begin{array}{c}
 \ACC_a^1\\
 \ACC_b^1\\
 \TOTAL^1
 \end{array}
 \right)
  =
 \left(
 \begin{array}{ccccc}
 1 & 0\\
 0 & 1\\
 \frac{1}{2} & \frac{1}{2}\\
 \end{array}
 \right)
 *
 \left(
 \begin{array}{c}
 u^1\\
 v^1\\
 \end{array}
    \right)
  \end{array}
  \]
  The rewritten formulas then again follow directly from the matrix.
  Repeating the application at all times yields:

  \vspace{.8em}
\begin{tabular}{|p{4em}|p{12em}|p{13em}|l}\hline  
  0 & 1 & 2 & \ldots \\ \hline
  $\LOAD_a^0$ &                        $\LOAD_a^0 + \LOAD_a^1 \equiv u^1$                                                  & $u^1 + \LOAD_a^2 \equiv u^2$                                                & \ldots \\
  $\LOAD_b^0$ &                        $\LOAD_b^0 + \LOAD_b^1 \equiv v^1$                                                  & $v^1 + \LOAD_b^2 \equiv v^2$                                                &  \ldots \\
  $\frac{\LOAD_a^0+\LOAD_b^0}{2}$ &    $\frac{(\LOAD_a^0+\LOAD_b^0)+(\LOAD_a^1+\LOAD_b^1)}{2} \equiv \frac{u^1 + v^1}{2}$  & $\frac{(u^1 + v^1)+(\LOAD_a^2+\LOAD_b^2)}{2} \equiv \frac{u^2 + v^2}{2}$    & \ldots \\ \hline
\end{tabular}
\vspace{.8em}

\noindent which results in a constant monitor.\qed

This pruning strategy can be generalized as well.
Let $\calR = \{r_1,\dots,r_m\}$ be a set of relevant variables (in
our case the output variables ${s}_i^t$) and
$\calV = \{s_1,\dots,s_n\} \cup \calR$ be the other variables (in our
case input variables or fresh variables from previous pruning
applications).
Let $\calC$ be a set of constraints which has to be fulfilled over
$r_1,\dots,r_m,s_1,\dots,s_n$, which contains equations of the form
$c = \sum_{i=1}^m c_{r_i} * r_i + \sum_{i=1}^n c_{s_i} * s_i
+ c'$ where $c,c',c_{s_i},c_{r_i}$ are constants.

If the equation system is unsolvable (which can easily be checked) we
return $\mathcal{C'} = \{0=1\}$, otherwise we can rewrite it as shown
on the left.
The matrix $N$ of this equation system has $m$ rows and $n$
columns.
Let $r$ be the rank of this matrix which is limited by $\min\{m,n\}$.
Consequently an $m \times r$ matrix $N'$ with $r \leq m$ exists with
the same span as $N$ and the system can be rewritten (without loosing
solutions to $(r_1,\dots,r_m)$).
From this rewritten equation system a new constraint
\begin{wrapfigure}[5]{l}{0.55\textwidth}
  \hspace{-0.3em}\begin{minipage}{0.58\textwidth}
    \vspace{-2.9em}
\[\begin{array}{l}
\left( \begin{array}{c} r_1\\ \vdots\\ r_m \end{array} \right)  =
 \left( \begin{array}{ccc}
 c_{1,1} & \dots  & c_{1,n}\\ & \vdots &\\ c_{m,1} & \dots  & c_{m,n}\\
 \end{array} \right)
*
 \left( \begin{array}{c} s_1\\ \vdots\\ s_n\\ \end{array} \right)
 +
 \left( \begin{array}{c} o_1\\ \vdots\\ o_m\\ \end{array}\right)
 \\[3em]
\end{array}\]
\end{minipage}
\end{wrapfigure}
set $\calC'$ can be generated which contains the equations from the
system.
We call this method the \LolaLA pruning strategy, which is perfect and
constant.

\newcounter{lem-lolaLAperfect}
\setcounter{lem-lolaLAperfect}{\value{lemma}}

\begin{lemma}
\label{lem:lolaLAperfect}
The $\LolaLA$ pruning strategy is perfect and constant.
\end{lemma}

  
%

\subsubsection{$\LolaBLA$}
%
%
%
  Consider the specification below (left) where $i$, $a$ and
  $b$ are input streams of type $\mathbb{R}$.
  Consider a trace where the values of stream $i$ are unknown until
  time $2$, but that we have the assumption $0 \leq i[\NOW] \leq 1$.
  The unpruned symbolic expressions describing the values of $x,y$ at
  time $2$ would then be in matrix notation:
  \[
    \begin{array}{l@{\hspace{4em}}l}
      \begin{array}{rcl}
        x & := & x[-1|0] + i[\NOW] \\
        y & := & 2 * y[-1|0] + i[\NOW] \\
        \OK & := & (a[\NOW] = x[\NOW]) \land (b[\NOW] = y[\NOW])
      \end{array} &
\left(
\begin{array}{c}
x^2\\
y^2
\end{array}
\right)
 =
\left(
\begin{array}{ccccc}
1 & 1 & 1\\
4 & 2 & 1\\
\end{array}
\right)
*
\left(
\begin{array}{c}
i^0\\
i^1 \\
i^2
\end{array}
      \right)
    \end{array}
  \]
\def\nudge{.5}
\tikzset{axis/.style={ultra thick, black, -latex, shorten <=-\nudge cm, shorten >=-2*\nudge cm}}
%
\begin{wrapfigure}[9]{r}{0.50\textwidth}
\raggedleft   
\vspace{-3.5em}
\begin{tikzpicture}[scale=0.7]

\tkzInit[xmax=8,ymax=3,xmin=-.5,ymin=-.5]
\tkzGrid
\tkzAxeXY[/tkzdrawX/label=$y^2$,/tkzdrawY/label=$x^2$]

\fill [blue!50] (0,0) -- (4,1) -- (6,2) -- (7,3) -- (3,2) -- (1,1) -- (0,0);

\draw [-stealth,line width=.7pt] (0,0) -- (4,1);
\draw [-stealth,line width=.7pt] (0,0) -- (2,1);
\draw [-stealth,line width=.7pt] (0,0) -- (1,1);
\end{tikzpicture}
\caption{Set of possible values of $x^2$ and $y^2$}
\label{fig:polygon}
\end{wrapfigure}
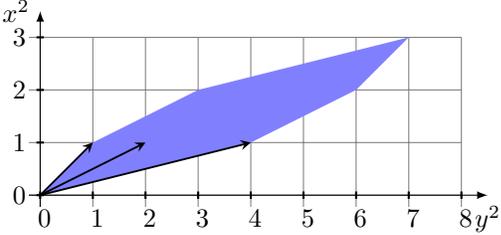
%
\noindent Since the assumption forces all $i^j$ to be
between $0$ and $1$ the possible set of value combinations $x$ and $y$
can take at time $2$ is described by a polyhedron with $6$ edges
depicted in Fig.~\ref{fig:polygon}.
%
Describing this polygon requires 3 vectors.
%
%
It is easy to see that each new unknown input generates a new vector,
which is not multiple of another. Hence for $n$ unknown
inputs on stream $i$ the set of possible value combinations for
$(x^t,y^t)$ is described by a polygon with $2n$ edges for which a constraint set of size 
$\mathcal{O}(n)$ is required. 
%
This counterexample implies that for \LolaBLA there is no perfect
pruning strategy.
However, one can apply the following approximation:
Given a constraint set $\calC$ over
$\calV = \{s_1,\dots,s_n\} \cup \calR$ with
$\calR = \{r_1,\dots,r_m\}$ the set of relevant variables.
\begin{compactenum}
\item Split the set of relevant variables into ${\R}_{\Bool}$
  containing those of type Boolean and ${\R}_{\Real}$ containing
  those of type Real.
\item For ${\R}_\Bool$ do the rewriting as for \LolaB obtaining
  $\calC'_\Bool$.
\item For ${\R}_\Real$ do the rewriting as for \LolaLA over
  $\calC^{\mathcal{LE}}$ with $\calC^{\mathcal{LE}} \subseteq \calC$
  being the set of all linear equations in $\calC$, obtaining
  $\calC'_\Real$.
\item For all fresh variables $v_i$ with $1 \leq i \leq k$ in
  $\calC'_\Real$ calculate a minimum bound $l_i$ and maximum bound $g_i$ (may be over-approximating) over
  the constraints $\calC \cup \calC'_\Real$ and build
  $\calC''_\Real = \calC'_\Real \cup \{l_i \leq v_i \leq g_i | 1 \leq
  i \leq k\}$.
 \item Return $\calC' = \calC'_\Bool \cup \calC''_\Real$
 \end{compactenum}
 We call this strategy the \LolaBLA pruning strategy, which allows to
 build an atemporal (imperfect but sound) constant-memory monitor.

\newcounter{lem-lolaBLA}
\setcounter{lem-lolaBLA}{\value{lemma}}
 
 \begin{lemma}
   \label{lem:lolaBLA}
  The \LolaBLA pruning strategy is sound and constant.
\end{lemma}

Note that with the \LolaBLA fragment we can also support if-then-else 
expressions.
A definition $s = ite(c,t,e)$ can be rewritten to handle $s$ as an
input stream adding assumption
$(c \land s = t) \lor (\neg c \land s = e)$.
After applying this strategy the specification is within the
$\LolaBLA$ fragment and as a consequence the sound (but imperfect)
pruning algorithms from there can be applied.

\subsection{Temporal assumptions}

We study now how to handle temporal assumptions.
Consider again Example~\ref{ex:loadPrune}, but instead of the assumption
$0 \leq ld[\NOW] \leq 10$ take
$0.9 * ld[-1,0] \leq ld[\NOW] \leq 1.1 * ld[-1,100]$.
In this case it would not be possible to apply the presented pruning
algorithms. In the pruning process at time 1 we would rewrite our
formulas in a fashion that they do not contain $ld^1$ anymore, but at time 2
we would receive the constraint 
$0.9 * ld^{1} \leq ld^2 \leq 1.1 * ld^{1}$
from the assumption.

Pruning strategies can be extended to consider variables which may be referenced by
input constraints at a later time as relevant variables, hence they
will not be pruned.
A monitor which receives constraint sets over the last $l$ instants is
called an $l$-lookback monitor.
An atemporal monitor is therefore a 0-lookback monitor.
For an $l$-lookback monitor the number of variables that are
referenced at a later timestamp is constant, so our pruning strategies
remain constant.
Hence, the following theorem is applicable to our pruning strategies
and as a consequence our solutions for atemporal monitors can be
adapted to $l$-lookback monitors (for constant $l$).
\begin{theorem}
   \label{th:perfectness2}
   Given a \Lola specification $\varphi$ and a constant pruning strategy $\calP$ for
   $\Expr^\Bool_\varphi$ there is a constant-memory $l$-lookback
   monitor $M_\varphi$ such that
   \begin{compactitem}
    \item $M_\varphi$ is sound if the pruning strategy is sound. 
    \item $M_\varphi$ is perfect if the pruning strategy is perfect.
   \end{compactitem}
\end{theorem}

\section{Implementation and Empirical Evaluation}
\label{sec:evaluation}

We have developed a prototype implementation of the symbolic
algorithm for past-only \Lola in Scala, using Z3~\cite{moura08Z3} as solver.
Our tool supports Reals and Booleans with their standard operations,
ranges (e.g. $[3,10.5]$) and $?$ for unknowns.
Assumptions can be encoded using the keyword
\texttt{ASSUMPTION}.\footnote{Note that for our symbolic approach
  assumptions can indeed be considered as a stream specification of
  type Boolean which has to be true at every time instant.}
Our tool performs pruning (Section~\ref{subsec:PruningLOLAFragments})
at every instant, printing precise outputs when possible. If an output
value is uncertain the formula and a range of possible values is
printed.
%

We evaluated two realistic case studies, a test drive data emission
monitoring~\cite{KohlHB18} and an electrocardiogram (ECG) peak
detector~\cite{GorostiagaS21}.
All measurements were done on a 64-bit Linux machine with an Intel
Core i7 and 8~GB RAM.
We measured the processing time of single events in our evaluation,
for inputs from 0 up to 20\% of uncertain values, resulting in
average of 25 ms per event (emissions case study) and 97 ms per event
(ECG).
In both cases the runtime per event did not depend on the length of
the trace (as predicted theoretically).
%
%
The longer runtime per event in the second case study is explained
because of the window of size 100 which is unrolled to 100 streams,
and using Z3 naively to deduce bounds of unknown variables.
We discuss the two case studies separately.

\subsubsection*{Case study \#1: Emission Monitoring}
The first example is a specification that receives test drive data
from a car (including speed, altitude, NOx emissions,\ldots)
from~\cite{KohlHB18}.
The \Lola specification is within $\LolaBLA$ (with $ite$), and checks
several properties, including \texttt{trip\_valid} which captures if
the trip was a valid test ride.
The specification contains around 50 stream definitions in total.
We used two real trips as inputs, one where the allowed NOx emission
was violated and one where the emission specification was satisfied.
%

%
We injected uncertainty into the two traces by randomly selecting
$x\%$ of the events and modifying the value within an interval 
$\pm y\%$.
The figure on the left shows the result of executing this experiment
for all integer combinations of $x$ and $y$
\begin{wrapfigure}[8]{l}{0.4\linewidth}
  \centering
\vspace{-2em}  
\includegraphics[width=1.05\linewidth]{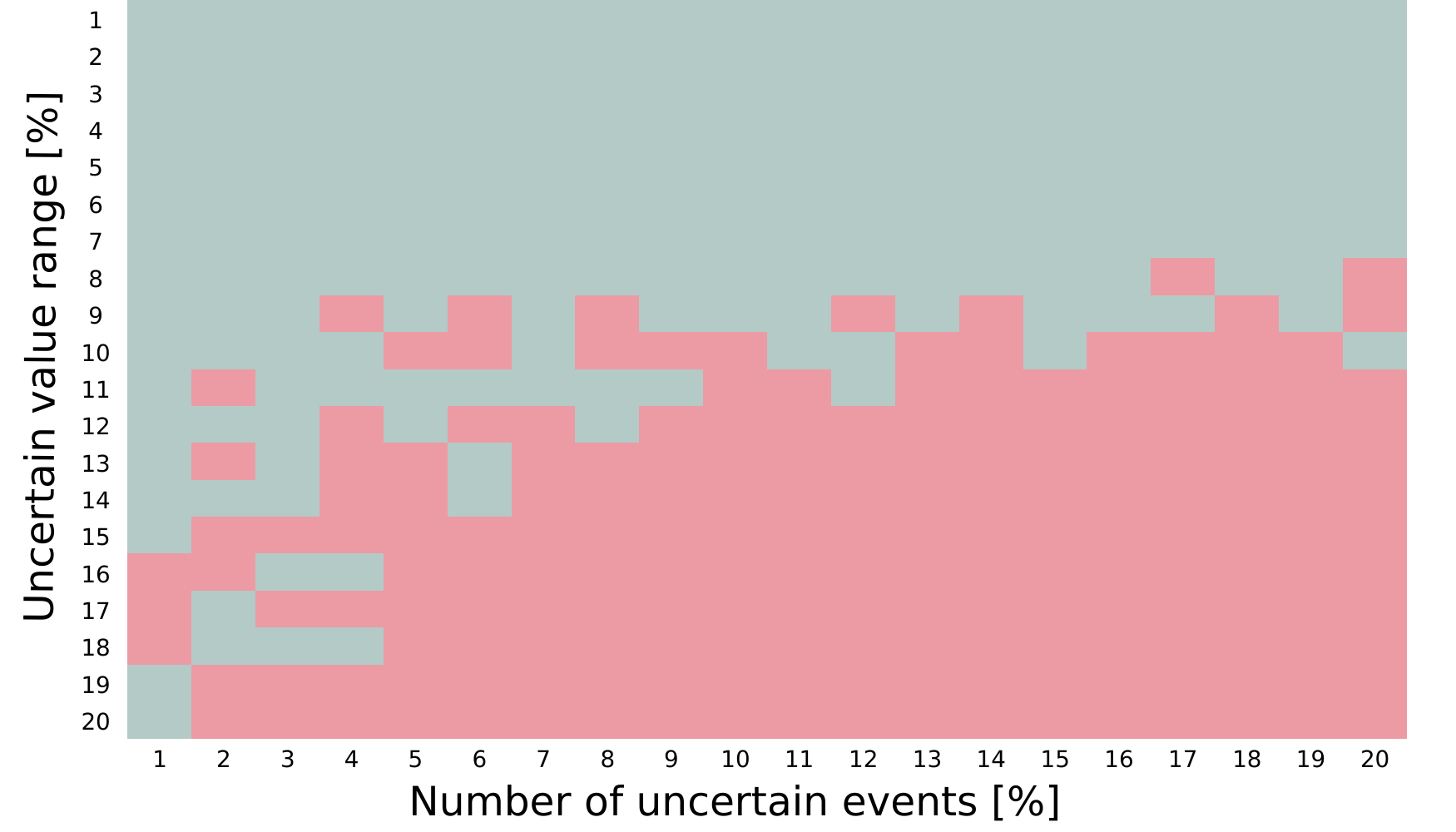}
\end{wrapfigure}
between $1$ and $20$, for one trace.
The green space represents the cases for which the monitor computed
the valid answer and the red space the cases where the monitor
reported unknown.
In both traces, even with $20\%$ of incorrect samples within an interval
of $\pm 7\%$ around the correct value the monitor was able to compute
the correct answer.
We also compared these results to the value-range approach, using
interval arithmetic. However, the final verdicts do not differ here. 
Though the symbolic approach is able to calculate more precise
intermediate results, these do not differ enough to obtain different
final Boolean verdicts.

%
As expected, for fully unknown values and no assumptions, neither the
symbolic nor the interval approaches could compute any certain
verdict, because the input values could be arbitrarily large.
%
However, in opposite to the interval approach, the symbolic approach allows adding 
assumptions (e.g.  the speed or altitude does not differ much from the previous value).
With this assumption, we received the valid result for
\texttt{trip\_valid} when up to 4\% of inputs are fully uncertain.
In other words, the capability of symbolic monitoring to encode
physical dependencies as assumptions often allows our technique to
compute correct verdicts in the presence of several unknown values.

\subsubsection*{Case Study \#2: Heart Rate monitoring}
Our second case study concerns the peak detection in electrocardiogram
(ECG) signals~\cite{GorostiagaS21}.
%
The specification calculates a sliding average and stores the values
of this convoluted stream in a window of size 100.
Then it checks if the central value is higher than the 50 previous and
the 50 next values to identifying a peak.

%
We evaluated the specification against a ECG trace with 2700 events
corresponding to 14 heartbeats.
We integrated uncertainty into the data in two different ways.
First, we modified $x\%$ percent of the events with deviations of
$\pm y\%$.
Even if $20\%$ of the values were modified with an error of
$\pm 20\%$, the symbolic approach returned the perfect result, while
the abstraction approach degraded over time because of accumulated
uncertainties (many peaks were incorrectly ``detected'', even under
$5\%$ of unknown values with a $\pm 20\%$ error---see front part of traces
in Fig.~\ref{fig:hr2}).
%
%
%
Second, we injected bursts of consecutive errors (? values) of
different lengths into the input data.
The interval domain approach lost track after the first burst and was
unable to recover, while the symbolic approach returned some ? around
the area with the bursts and recovered when new values were received
(see Fig.~\ref{fig:hr2}).

\begin{figure}[t!]
\centering
\begin{minipage}{.5\textwidth}
  \centering
  \includegraphics[width=0.90\linewidth,height=1.8cm]{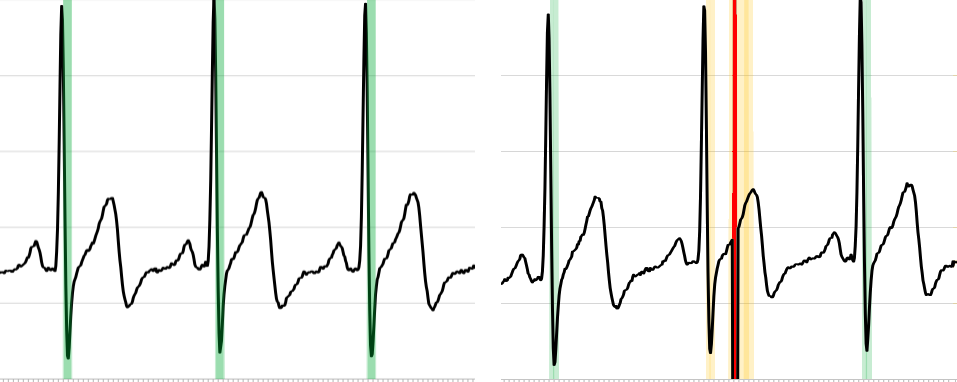}
\end{minipage}%
\begin{minipage}{.5\textwidth}
  \centering
  \includegraphics[width=0.90\linewidth,height=1.8cm]{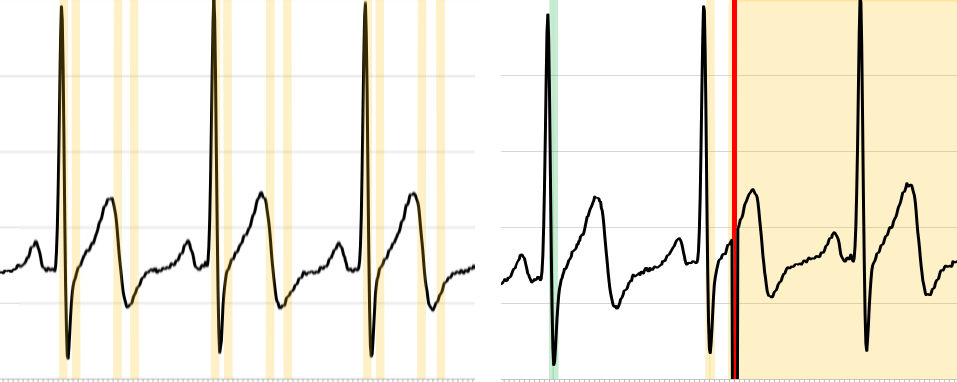}
\end{minipage}
\caption{ECG analysis. Left: Symbolic approach, Right: Value range approach.
Green: Certain heartbeats, Yellow: Potential heartbeats, Red: Bursts of unknown
values.}
\label{fig:hr2}
\end{figure}

%
We exploited the ability of symbolic monitors to handle assumptions by
encoding that heartbeats must be apart from each other more than 160
steps (roughly $0.5$ seconds), which increased the accuracy.
In one example (Fig.~\ref{fig:hr3} in appendix \ref{sec:GraphsAppx})
the monitor correctly detected a peak right after a burst of errors.
The assumption allows the monitor to infer that the unknown burst of
values are below a certain threshold, which enables the detection of
the next heartbeat.
If instead of the assumption heartbeats that are not at least 160
steps apart were simply filtered out, future heartbeats could not be
detected correctly (Fig.~\ref{fig:hr3}(b) in appendix
\ref{sec:GraphsAppx}), because no ranges of the values of unknown
events can be deduced.


%

\vspace{-0.5em}\section{Conclusion}
\label{sec:conclusion}


We have introduced the concept of symbolic monitoring to monitor in the presence of
input uncertainties and assumptions on the system behavior.
We showed theoretically and empirically that symbolic monitoring
is more precise than a straightforward abstract interpretation
approach, and have identified logical theories in which perfect
symbolic approach can be implemented efficiently (constant
monitoring).
Future work includes: (1) to identify other logical theories and their
combinations that guarantee perfect trace length independent
monitoring;
(2) to be able to anticipate verdicts ahead of time for rich data
domains by unfolding the symbolic representation of the specification
beyond, along the lines of
\cite{DBLP:conf/fsttcs/BauerLS06,DBLP:conf/rv/Leucker12,DBLP:conf/nfm/ZhangLD12}
for Booleans;
%
%

Finally, we envision that symbolic monitoring can become a general,
foundational approach for monitoring that will allow to explain many
existing monitoring approaches as instances of the general schema.



\savegeometry{lncs}
 \newgeometry{text={400bp,610bp},
    tmargin=100bp,
    lmargin=112bp, 
   rmargin=112bp, 
   marginparwidth=72bp, 
   marginparsep=12bp}

\bibliographystyle{splncs04}
\bibliography{short}

\addtolength{\textwidth}{-1em}

\loadgeometry{lncs}

\vfill
\pagebreak
\appendix
\section{Further Examples}
\label{app:examples}

\begin{example}
  Consider again $\varphi$ in Example~\ref{ex:one}, with the prefect
  readings in Fig.~\ref{fig:exone}(b).
  We show for time steps $0$, $1$, $2$ an $3$ the equations before and
  after evaluation and simplification, in the upper and lower rows,
  resp.
  \noindent%
  
  \[
    \footnotesize
    \begin{array}{|l|l|l|l|}\hline
      0 & 1 & 2 & 3 \\ \hline
      \LOAD^0=3 & \LOAD^1=4 & \LOAD^2=5 & \LOAD^2=7 \\
      \ACC^0=\LOAD^0 & \ACC^1=\ACC^0+\LOAD^1 &\ACC^2=\ACC^1+\LOAD^2 &\ACC^3=\ACC^2+\LOAD^3-\LOAD^0 \\
      \OK^0=\ACC^0\leq 15 & \OK^1=\ACC^1\leq 15 & \OK^2=\ACC^2\leq 15 & \OK^3=\ACC^3\leq 15 \\ \hline
      & \LOAD^0=3 & \LOAD^0=3,\LOAD^1=4 & \LOAD^1=4,\LOAD^2=5 \\
      \LOAD^0=3 & \LOAD^1=4 & \LOAD^2=5 & \LOAD^3=7 \\ 
      \ACC^0=3 & \ACC^1=7 & \ACC^2=12 & \ACC^3=16 \\
      \OK^0=\TRUE & \OK^1=\TRUE & \OK^2=\TRUE & \OK^3=\FALSE \\ \hline
  \end{array}
\]
All equations are fully resolved at every step.
Also, $\LOAD^0$ is pruned at 3 because $\LOAD^0$
will not be used in the future.
Consider now the imperfect input in~Fig.\ref{fig:exone}(d):
\[
  \footnotesize
    \begin{array}{|l|l|l|l|}\hline
      0 & 1 & 2 & 3 \\ \hline
      1\leq\LOAD^0\leq{}5 & \LOAD^1=4 & \LOAD^2=5 & \LOAD^2=7 \\
      \ACC^0=\LOAD^0 & \ACC^1=\ACC^0+\LOAD^1 &\ACC^2=\ACC^1+\LOAD^2 &\ACC^3=\ACC^2+\LOAD^3-\LOAD^0 \\
      \OK^0=\ACC^0\leq 15 & \OK^1=\ACC^1\leq 15 & \OK^2=\ACC^2\leq 15 & \OK^3=\ACC^3\leq 15 \\ \hline
      1\leq\LOAD^0\leq{}5 & \LOAD^1=4 & \LOAD^2=5 & \LOAD^3=7 \\ 
      & 1\leq\LOAD^0\leq{}5 & 1\leq\LOAD^0\leq{}5,\LOAD^1=4 & \LOAD^1=4,\LOAD^2=5\\
      \ACC^0=\LOAD^0 & \ACC^1=\LOAD^0+4 & \ACC^2=\LOAD^0+9 & \ACC^3=16 \\
      \OK^0=\TRUE & \OK^1=\TRUE & \OK^2=\TRUE & \OK^3=\FALSE \\ \hline
  \end{array}
\]
At time $2$, $\OK^2=\TRUE$ is inferred from
$\{1\leq\LOAD^0\leq{}5,\LOAD^1=4,\ACC^2=\LOAD^0+9,\OK^2=\ACC^2\leq 15\}$.
At time $3$ the dependency to the unknown value $\LOAD^0$ is
eliminated from $\ACC^3=\ACC^2+\LOAD^2-\ACC^0$ by symbolic
manipulation.\qed
\end{example}

\section{Missing Proofs}
\label{app:proofs}

\newcounter{backup}

\setcounter{backup}{\value{theorem}}
\setcounter{theorem}{\value{thm-perfect}}

  \begin{theorem}
   Given a specification $\varphi$ and a constant pruning strategy
   $\calP$ for $\Expr^\Bool_\varphi$, there is an atemporal
   constant-memory monitor $M_\varphi$ s.t.
   \begin{compactitem}
    \item $M_\varphi$ is sound if the pruning strategy is sound. 
    \item $M_\varphi$ is perfect if the pruning strategy is perfect.
   \end{compactitem}
  \end{theorem}

  \begin{proof}
    Let $\varphi = (I,O,E)$ be the atemporal specification.\\\\
    We will show the theorem by constructing a monitor that generates expression sets which satisfy the relation to $\Phi^i$ demanded in \Cref{def:SoundPerfectMon}.\\\\
    In general a monitor $M$ is perfect with respect to $\varphi$, if it generates $M^i$, s.t.
    $\sem{\Phi^i}_{{\R}^i} = \sem{M^{i}}_{{\R}^i}$ and
    sound, if it generates $M'^i$, s.t.
    $\sem{\Phi^i}_{{\R}^i} \subseteq \sem{M^{i}}_{{\R}^i}$ with
    ${\R}^i = \{x^i | x \in I \cup O\}$, the relevant variables for this timestamp. 
    This follows directly from Definition~\ref{def:SoundPerfectMon} and Definition~\ref{def:pruningStrat}.\\\\
    We first consider the case where the pruning strategy $\calP$ (we write $\calP^i$ for the pruning according to ${\R}^i$) is sound:\\\\
    By using $\calP$ we can construct a sound monitor $M$ which generates outputs 
    $M^0=\calP^0(\semSym{\varphi}^0 \cup \sem{A(0)}_\varphi \cup \psi^{0})$ and
    $M^i=\calP^i(M^{i-1} \cup \semSym{\varphi}^i \cup \sem{A^i}_\varphi \cup \psi^{i})$ for $i > 0$.\\
    We will now show $\sem{\Phi^{i}}_{{\R}^i} \subseteq \sem{M^{i}}_{{\R}^i}$ 
    for all $i$, i.e. that $M$ is sound. 
    Afterwards we will argue $M$ is also a constant-memory monitor.\\\\
    For $i=0$:\\
    By Definition~\ref{def:pruningStrat}:
    $\sem{\Phi^0}_{{\R}^0} \subseteq \sem{\calP^0(\Phi^0)}_{{\R}^0}$.\\
    And by Definition of $\Phi^0$:
    $\sem{\calP^0(\Phi^0)}_{{\R}^0} = \sem{\calP(\semSym{\varphi}^0 \cup \sem{A(0)}_\varphi \cup \psi^{0})}_{{\R}^0} = \sem{M^0}_{{\R}^0}$.\\\\
    Further, for $i > 0$:\\
    $\sem{\Phi^i}_{{\R}^i} =
    \sem{\Phi^{i-1} \cup \semSym{\varphi}^i \cup \sem{A^i}_\varphi \cup \psi^i}_{{\R}^i} \subseteq
    \sem{M^{i-1} \cup \semSym{\varphi}^i \cup \sem{A^i}_\varphi \cup \psi^i}_{{\R}^i}$.\\
    This is because of the atemporality of the monitor and the flattened form of the specification all common variables of $\Phi^{i-1}$ and 
    $\semSym{\varphi}^i \cup \sem{A^i}_\varphi \cup \psi^i$ are from ${\R}^{i-1}$ for which
    $\sem{M^{i-1}}_{{\R}^{i-1}}$ is known to be a superset of  $\sem{\Phi^{i-1}}_{{\R}^{i-1}}$.
    Hence if for any $\alpha$ we have 
    $\alpha \models \semSym{\varphi}^i \cup \sem{A^i}_\varphi \cup \psi^i$ and 
    $\alpha \models \Phi^{i-1}$ 
    then also 
    $\alpha \models M^{i-1}$.
    Thus, if
    $\alpha \models \Phi^{i-1} \cup \semSym{\varphi}^i \cup \sem{A^i}_\varphi \cup \psi^i$ 
    then
    $\alpha \models M^{i-1} \cup \semSym{\varphi}^i \cup \sem{A^i}_\varphi \cup \psi^i$.
    This fact together with the definition of $\sem{\cdot}_{{\R}^i}$ (Definition~\ref{def:pruningStrat}) implies the subset relation above.\\
    Furthermore
    $\sem{\Phi^i}_{{\R}^i} \subseteq
     \sem{\calP^i(M^{i-1} \cup \semSym{\varphi}^i \cup \sem{A^i}_\varphi \cup \psi^i)}_{{\R}^i} = 
     \sem{M^i}_{{\R}^i}$ again by
    definition of $\calP^i$ (Definition~\ref{def:pruningStrat}).\\
    Hence, for all outputs of $M$ we have
    $\sem{\Phi^{i}}_{{\R}^i} \subseteq \sem{M^{i}}_{{\R}^i}$ for all $i$ and thus $M$ is sound.\\\\
    Note that for all $M^{i}$ we have $|M^{i}| \leq c$ due to the  constant pruning strategy. However $M$ only has to store the $M^{i}$ 
    from the last step and as consequence it is constant-memory monitor.\\\\
    The case where the pruning strategy $\calP$ is perfect is analogous.
  \qed
\end{proof}

\setcounter{theorem}{\value{backup}}

\setcounter{backup}{\value{lemma}}
\setcounter{lemma}{\value{lem-lolaBperfect}}

\begin{lemma}
  The  $\LolaB$ pruning strategy is perfect and constant.
\end{lemma}

\begin{proof}
  Let $\calC$ be any constraint set over
  $\mathcal{V} = \{s_1,\dots,s_n\} \cup {\R}$ with relevant variables
  $\calR = \{r_1,\dots,r_m\}$ and $\calC'$ the set obtained after
  pruning.\\\\
  Clearly $\semR{\calC} = \semR{\calC'}$ since by definition
  $\semR{\calC} = \{(v_1,\dots,v_n) | (r_1 = v_1)
  \land \dots \land (r_n = v_n) \models \calC \}$.
  We only add value combinations to $T$ which fulfill
  $(r_1 = v_1) \land \dots \land (r_m = v_m) \models \gamma$ and the
  $\psi_i$ in $\calC'$ by definition just allow exactly these
  combinations.\\\\
  Moreover, the value table from which $\calC'$ is created has $m$
  rows and $c$ columns.\\
  Hence $\calC'$ contains $m$ formulas over at most $\lceil \log(c) \rceil$
  variables.\\
  According to our measure we have for every $r_i = \psi_i$ from
  $\calC'$
  $|r_i = \psi_i| = 1 + |\psi_i| \leq 1 + \lceil \log(c) \rceil * 2^{\lceil \log(c) \rceil} \leq
  c^2 + 1$ and consequently $|\calC'| \leq m * (c^2 + 1)$.\\
  Note that for $c$ we have $c \leq 2^m$ (number of columns in the
  table) and $m$ is the number of streams in the flattened
  specification and hence constant.\qed
\end{proof}

\setcounter{lemma}{\value{backup}}

\setcounter{backup}{\value{lemma}}
\setcounter{lemma}{\value{lem-lolaLAperfect}}

\begin{lemma}
The $\LolaLA$ pruning strategy is perfect and constant.
\end{lemma}

\begin{proof}
  Let $\calC$ be any constraint set over
  $\calV = \{s_1,\dots,s_n\} \cup {\R}$ with relevant variables
  $\calR = \{r_1,\dots,r_m\}$ and $\calC'$ the set obtained after
  pruning.\\\\
  The strategy is clearly perfect.
  If $\calC$ did not have solutions we return a $\calC'$ which also
  has no solutions.
  If $\calC$ has solutions we use equivalence transformations of the
  system of equations preserving the solutions for $(r_1,\dots,r_m)$,
  hence
  $\semR{\calC} = \{(v_1,\dots,v_n) | (r_1 = v_1)
  \land \dots \land (r_n = v_n) \models \calC \} = \semR{\calC'}$.\\\\
  Note that $N'$ is an $m \times r$ matrix with $r \leq m$.
  Hence in $\calC'$ there are $m$ expressions of the form
  $r_i = \sum_{j=1}^r c_{i,j} v_j + c_i$ with
  $|r_i = \sum_{j=1}^r c_{i,j} v_j + c_i| = 2r + 2$ and hence
  $|\calC'| = m * (2r + 2) \leq 2m^2 + 2m$.
  The value $m$ is the number of streams in the flattened
  specification and hence constant.\qed
\end{proof}

\setcounter{lemma}{\value{backup}}

\setcounter{backup}{\value{lemma}}
\setcounter{lemma}{\value{lem-lolaBLA}}

 \begin{lemma}
  The \LolaBLA pruning strategy is sound and constant.
\end{lemma}

\begin{proof}
  Let $\calC$ be any constraint set over
  $\calV = \{s_1,\dots,s_n\} \cup {\R}$ with relevant variables
  $\R = \{r_1,\dots,r_m\}$ and $\calC'$ the set obtained after pruning.\\\\
  It follows from Lemma~\ref{lem:lolaBperfect} that $\calC'_\Bool$
  is a perfect pruning of $\calC$ for ${\R}_\Bool$ and from
  Lemma~\ref{lem:lolaLAperfect} that $\calC'_\Real$ is a
  perfect pruning of $\calC^\mathcal{LE}$ for ${\R}_\Real$.\\
  Since
  $\calC \cup \calC'_\Real \models \{l_i \leq v_i \leq g_i | 1 \leq i
  \leq k\}$ by definition of $l_i, g_i$, we have:
  $\semR{\calC'} = \sem{\calC'_\Bool \cup \calC''_\Real}_{\R} =
  \sem{\calC'_\Bool \cup \calC'_\Real \cup \{l_i \leq v_i \leq g_i | 1 \leq i \leq k\}}_{\R}
  \supseteq \semR{\calC'_\Bool \cup \calC'_\Real \cup \calC}$.\\\\
  We also have that, $\calC'_\Bool$ and $\calC$ only share the variables ${\R}_\Bool$ 
  (because our pruning strategy for \LolaB pruned all others away and only introduced fresh variables). 
  Furthermore $\sem{\calC'_\Bool}_{\R_\Bool} = \sem{\calC}_{\R_\Bool}$.
  Thus it follows for all expressions $\alpha$ over $\R$, that if 
  $\alpha \models \calC$ then $\alpha \models \calC'_\Bool$.\\
  The same reasoning holds for $\calC'_\Real$ and $\calC$, i.e. 
  $\alpha \models \calC$ then $\alpha \models \calC'_\Real$ for all expressions $\alpha$ over $\R$.\\ 
  Hence, it follows:
  $\semR{\calC'_\Bool \cup \calC'_\Real \cup \calC} \supseteq \semR{\calC}$: 
  The \LolaBLA pruning strategy is sound. \\\\
  Moreover, the pruning strategy is also constant:
  From Lemmas~\ref{lem:lolaBperfect} and \ref{lem:lolaLAperfect} ,
  $\calC'_\Bool$ and $\calC'_\Real$ have constant upper
  bounds.
  The number of fixed-sized constraints added in step 3 is bounded by
  the number of fresh variables in the \LolaLA pruning strategy which
  is bounded by $m$. The value $m$ is the number of streams in the flattened
  specification and hence constant.
  Thus the size of $\calC'$ has a constant upper bound as well.\qed
\end{proof}

\setcounter{lemma}{\value{backup}}

\section{Graphs from evaluation runs}
\label{sec:GraphsAppx}
\hannestodo{Feel free to remove some figures again if you do not like them}

\newcommand{\nonfloatcaption}[1]{\par\ \\\raggedright\textbf{Figure} {#1}}

\subsection{Emission Monitoring: Verdicts for uncertain inputs}

\begin{minipage}{0.99\textwidth}
\centering
 \begin{minipage}{.5\textwidth}
   \centering
    \includegraphics[width=\linewidth]{png-plots/emiss_empty_ov.pdf}
 \end{minipage}%
 \begin{minipage}{.5\textwidth}
   \centering
   \includegraphics[width=.95\linewidth]{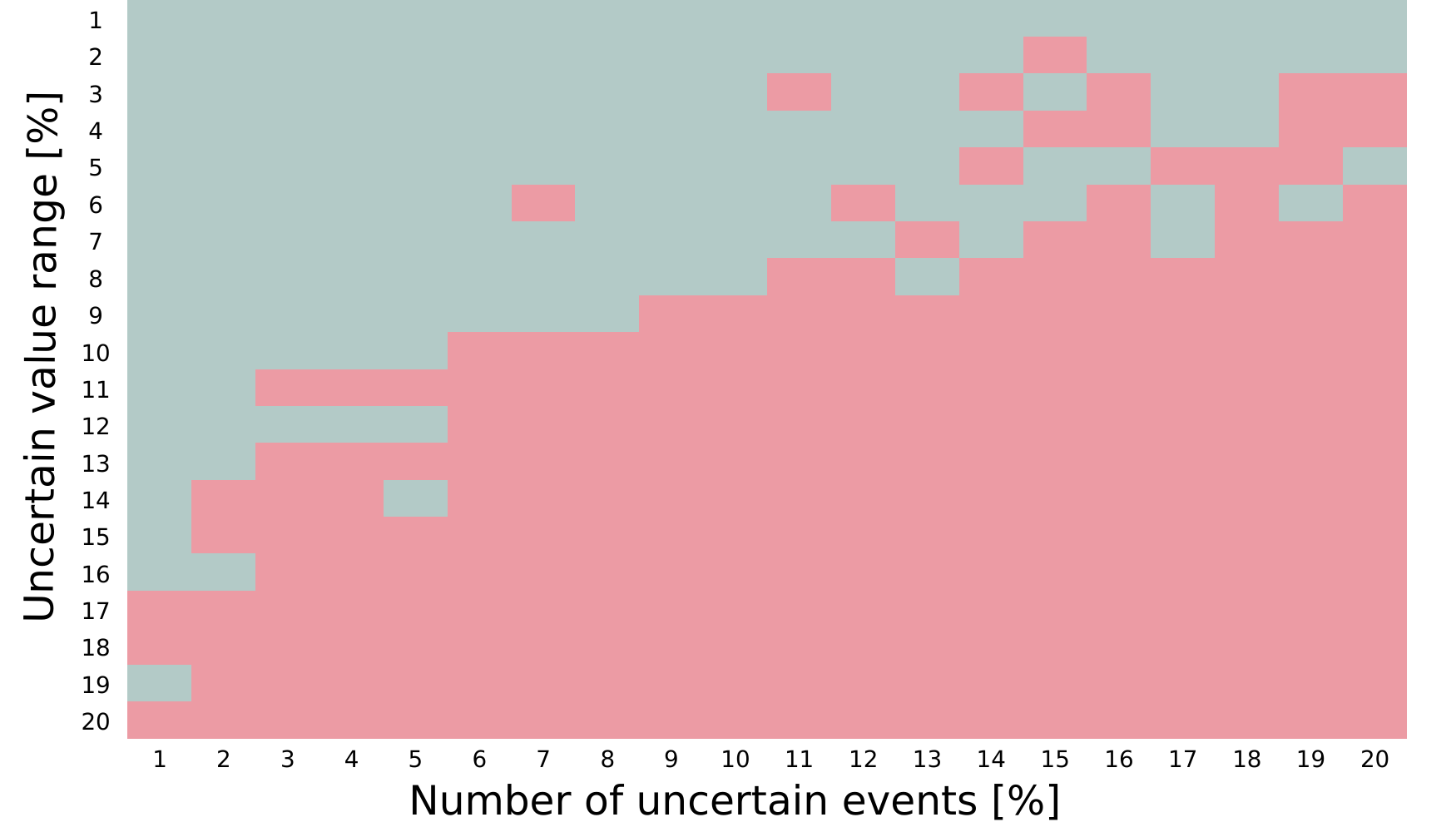}
 \end{minipage}
\nonfloatcaption{Monitoring \texttt{trip\_valid} for two traces with  1\% to 20\% of values uncertain (range of $\pm 1\%$ to $\pm 20\% $ around correct value). 
 Green: certain result; red: uncertain result.}
\vspace{-1em}
\end{minipage}

\subsection{ECG (symbolic), Full run for uncertain inputs}

\begin{minipage}{0.99\textwidth}
\centering
\includegraphics[width=0.95\linewidth]{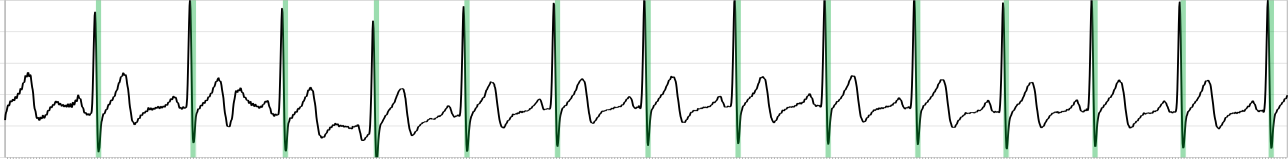}

\nonfloatcaption{ECG analysis 20\% of the values uncertain (range of $\pm 20\%$ around correct value). 
  Symbolic approach. 
  Green: heartbeat certainly detected; yellow: heartbeat possibly detected.}
\end{minipage}

\subsection{ECG (intervals), Full run for uncertain inputs}

\begin{minipage}{0.99\textwidth}
\centering
\includegraphics[width=0.95\linewidth]{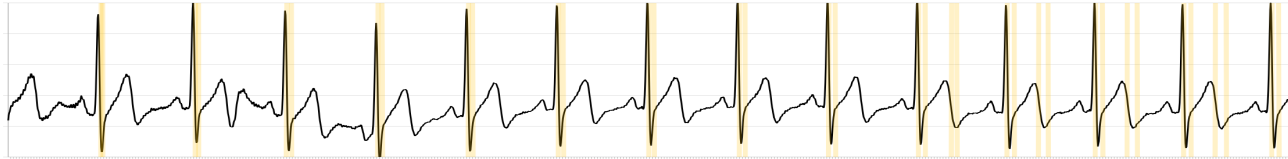}

\nonfloatcaption{ECG analysis 5\% of the values uncertain (range of $\pm 20\%$ around correct value).
  Interval approach.
  Green: heartbeat certainly detected; yellow: heartbeat possibly detected.}
\end{minipage}

\subsection{ECG (symbolic), Full run for 5 uncertainty bursts}

\begin{minipage}{0.99\textwidth}
\centering
\includegraphics[width=0.95\linewidth]{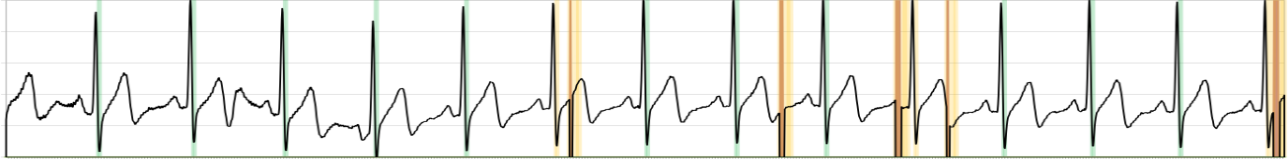}

\nonfloatcaption{ECG analysis with 5 bursts of fully uncertain values(5 to 20 in a row). 
 Symbolic approach.
 Orange: Burst of uncertain values; green: heartbeat certainly detected; yellow: heartbeat possibly detected.}
\end{minipage}

\subsection{ECG (intervals), Full run for 5 uncertainty bursts}

\begin{minipage}{0.99\textwidth}
\centering
\includegraphics[width=0.95\linewidth]{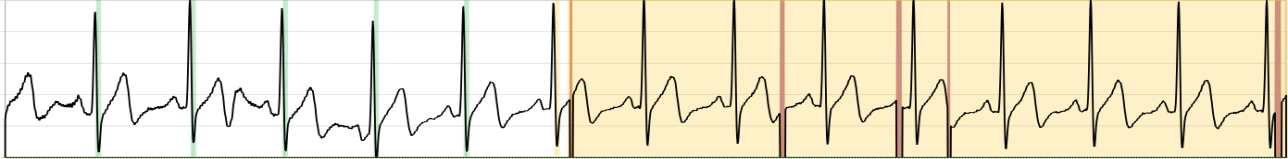}

\nonfloatcaption{ECG analysis with 5 bursts of fully uncertain values(5 to 20 in a row). 
 Interval approach.
 Orange: Burst of uncertain values; green: heartbeat certainly detected; yellow: heartbeat possibly detected.}
\end{minipage}

\subsection{ECG (symbolic), Full run for 5 uncertainty bursts (with assumption)}

\begin{minipage}{0.99\textwidth}
\centering
\includegraphics[width=0.95\linewidth]{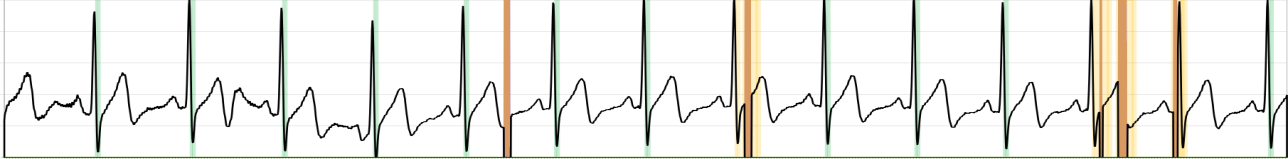}

\nonfloatcaption{ECG analysis with 5 bursts of fully uncertain values(5 to 20 in a row).
 Symbolic approach with assumption.
 Orange: Burst of uncertain values, green: heartbeat certainly detected, yellow: heartbeat possibly detected.}
\end{minipage}

\subsection{ECG (symbolic), Full run for 5 uncertainty bursts (with filter)}

\begin{minipage}{0.99\textwidth}
\centering
\includegraphics[width=0.95\linewidth]{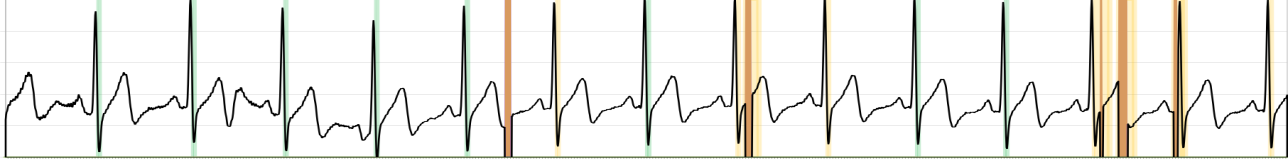}

\nonfloatcaption{ECG analysis with 5 bursts of fully uncertain values(5 to 20 in a row). 
  Symbolic approach with filter.
  Orange: Burst of uncertain values; green: heartbeat certainly detected; yellow: heartbeat possibly detected.}
\end{minipage}

\subsection{ECG (symbolic), Bursts: Comparison Assumptions and Filter}

\begin{minipage}{0.99\textwidth}
\centering
\begin{minipage}{.5\textwidth}
  \centering
  \includegraphics[width=0.95\linewidth]{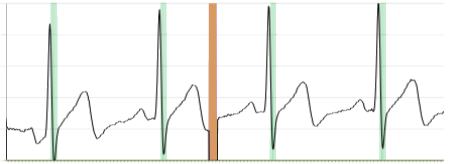}
\end{minipage}%
\begin{minipage}{.5\textwidth}
  \centering
  \includegraphics[width=0.95\linewidth]{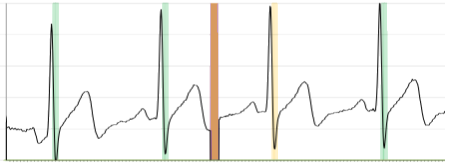}
\end{minipage}
\nonfloatcaption{ECG analysis with bursts. 
  Left: Usage of assumption. Right: Additional condition added to output stream.
  Orange: Burst of uncertain values; green: heartbeat certainly detected; yellow: heartbeat possibly detected.}
\label{fig:hr3}
\end{minipage}


\end{document}